\documentclass[journal]{IEEEtran}
\ifCLASSINFOpdf
\else
\fi

\usepackage{graphics} % for pdf, bitmapped graphics files
\usepackage{graphicx}
\usepackage{times} % assumes new font selection scheme installed
\usepackage{amsmath} % assumes amsmath package installed
\usepackage{amssymb}  % assumes amsmath package installed
\usepackage{algorithm}
\usepackage{color}

\usepackage{verbatim}

\usepackage{amsthm}

\usepackage{mathtools}

\usepackage{subfigure}

\newtheorem{theorem}{\bf Theorem} \newtheorem{definition}{\bf Definition} 
\newtheorem{lemma}{\bf Lemma} \newtheorem{remark}{\bf Remark}
  \newtheorem{proposition}{\bf Proposition} 
\newtheorem{assumption}{\bf Assumption}  
\newtheorem{Algorithm}{\bf Algorithm}

\begin{document}
\title{Data-Driven Model Predictive Control with Stability and Robustness Guarantees}

\author{Julian Berberich$^1$, %~\IEEEmembership{Member,~IEEE,}
		Johannes K\"ohler$^1$, %~\IEEEmembership{Fellow,~OSA,}
		Matthias A. M\"uller$^2$, %~\IEEEmembership{Life~Fellow,~IEEE}
		and Frank Allg\"ower$^1$.% <-this % stops a space
		\thanks{This work was funded by Deutsche Forschungsgemeinschaft (DFG,
German Research Foundation) under Germany’s Excellence Strategy - EXC
2075 - 390740016. The authors thank the International Max Planck Research
School for Intelligent Systems (IMPRS-IS) for supporting Julian Berberich,
and the International Research Training Group Soft Tissue Robotics (GRK
2198/1).}
\thanks{$^1$Julian Berberich, Johannes K\"ohler, and Frank Allg\"ower are with the Institute for Systems Theory and Automatic Control, University of Stuttgart, 70550 Stuttgart, Germany (email:$\{$ julian.berberich, johannes.koehler, frank.allgower$\}$@ist.uni-stuttgart.de)}
\thanks{$^2$Matthias A. M\"uller is with the Leibniz University Hannover, Institute of Automatic Control, 30167 Hannover, Germany (e-mail:mueller@irt.uni-hannover.de)}}% <-this % stops a space


\markboth{}%
{}



\IEEEoverridecommandlockouts

\IEEEpubid{\begin{minipage}{\textwidth}\ \\[12pt] \\ \\
\copyright 2020 IEEE. Personal use of this material is permitted. Permission from IEEE must be obtained for all other uses, in any current or future media, including reprinting/republishing this material for advertising or promotional purposes, creating new collective works, for resale or redistribution to servers or lists, or reuse of any copyrighted component of this work in other works.
\end{minipage}}

\maketitle

\begin{abstract}
We propose a robust data-driven model predictive control (MPC) scheme to control linear time-invariant (LTI) systems.
The scheme uses an implicit model description based on behavioral systems theory and past measured trajectories.
In particular, it does not require any prior identification step, but only an initially measured input-output trajectory as well as an upper bound on the order of the unknown system.
First, we prove exponential stability of a nominal data-driven MPC scheme with terminal equality constraints in the case of no measurement noise.
For bounded additive output measurement noise, we propose a robust modification of the scheme, including a slack variable with regularization in the cost.
We prove that the application of this robust MPC scheme in a multi-step fashion leads to practical exponential stability of the closed loop w.r.t. the noise level.
%, where the region of attraction approaches the set of all feasible initial conditions.
%This requires an appropriate selection of the regularization parameters, a low noise level, as well as a sufficiently large persistence of excitation relative to the noise amplitude.
The presented results provide the first (theoretical) analysis of closed-loop properties, resulting from a simple, purely data-driven MPC scheme.
\end{abstract}


% \fi
\IEEEpeerreviewmaketitle

\section{Introduction}
%\JK{\cite{rosolia2018learning}: hat ein Model, aber verbessert performance online, nicht klassisches Data-Driven/Model-free (kann man wegen mir auch weglassen).}\\
While data-driven methods for system analysis and control have become increasingly popular over the recent years, only few such methods give theoretical guarantees on, e.g., stability or constraint satisfaction of system variables~\cite{Hou13,Recht18}.
A control method, which is naturally well-suited for achieving these objectives is model predictive control (MPC), which can handle nonlinear system dynamics, hard constraints on input, state and output, and it takes performance criteria into account~\cite{Rawlings09}.
It centers around the repeated online solution of an optimization problem over predicted future system trajectories.
Thus, for the implementation of MPC, a model of the plant is required, which is usually obtained from first principles or from measured data via system identification~\cite{Ljung87}.
An appealing alternative is to implement an MPC controller directly from measured data, without prior knowledge of an accurate model.
In various recent works, learning-based or adaptive MPC schemes have been proposed, which improve an inaccurate initial model using online measurements~\cite{adetola2011robust,Aswani13,tanaskovic2014adaptive,Berkenkamp17,Zanon19}, while giving guarantees on the resulting closed loop.
Similarly, MPC based on Gaussian Processes has received increasing attraction~\cite{Hewing18}, but proving desirable closed-loop properties remains an open issue.
A different approach, which uses linear combinations of past trajectories to predict future trajectories, has been presented in~\cite{salvador2018data}, but also no guarantees on, e.g., stability of the closed loop were given.
The design of purely data-driven MPC approaches with guarantees on stability and constraint satisfaction thus remains an open problem.

In this paper, we present a novel data-driven MPC scheme to control linear time-invariant (LTI) systems with stability and robustness guarantees for the closed loop.
Our approach relies on a result from behavioral systems theory, which shows that the Hankel matrix consisting of a previously measured input-output trajectory spans the vector space of all trajectories of an LTI system, given that the input component is persistently exciting~\cite{Willems05}.
Although this result has found various applications in the field of system identification~\cite{Katayama05,Markovsky05,Markovsky08}, it has only recently been used to develop data-driven methods for system analysis and control with theoretical guarantees.
An exposition of the main result of~\cite{Willems05} in the classical state-space control framework and an extension to certain classes of nonlinear systems are provided in~\cite{Berberich19}.
Further, the result is employed in~\cite{Persis19} to design state- and output-feedback controllers and in~\cite{Romer19} to verify dissipation inequalities from measured data, whereas~\cite{Waarde19} investigates data-driven control without requiring persistently exciting data.

Moreover, the recent contributions~\cite{Yang15,Coulson19,Coulson19b} set up an MPC scheme based on~\cite{Willems05}, but no guarantees on recursive feasibility or closed-loop stability can be given since neither terminal ingredients are included in the MPC scheme nor sufficient lower bounds on the prediction horizon are derived.
In the present paper, we propose a related MPC scheme, which utilizes terminal equality constraints, and we provide a theoretical analysis of various desirable properties of the closed loop.
To the best of our knowledge, this is the first analysis regarding recursive feasibility and stability
of purely data-driven MPC.
The main advantage of the proposed MPC scheme over existing adaptive or learning-based methods such as~\cite{adetola2011robust,Aswani13,tanaskovic2014adaptive,Berkenkamp17,Zanon19} is that it requires only an initially measured, persistently exciting data trajectory as well as an upper bound on the system order, but no (set-based) model description and no online estimation process.
Moreover, since it relies on the data-driven system description from~\cite{Willems05}, the presented scheme is inherently an output-feedback MPC scheme and does not require online state measurements.\newpage

After stating the required definitions and existing results in Section~\ref{sec:setting}, we expand the nominal MPC scheme of~\cite{Yang15,Coulson19} by terminal equality constraints in Section~\ref{sec:tec}.
Under the assumption that the output of the plant can be measured exactly, we prove recursive feasibility, constraint satisfaction, and exponential stability of the scheme.
In Section~\ref{sec:robust}, we propose a robust data-driven MPC scheme to account for bounded additive noise in both the initial data for prediction as well as the online measurements.
%Throughout Section~\ref{sec:robust}, we do not consider output constraints, but we plan to address them in future research.
%The scheme includes  output constraints, which ensure that the closed-loop output satisfies the constraints despite the noisy data.
Under suitable assumptions on the system and design parameters, we prove that the closed loop under application of the scheme in a multi-step fashion leads to a practically exponentially stable closed loop.
In Section~\ref{sec:example}, we illustrate the advantages of the proposed scheme over the scheme without terminal constraints from~\cite{Yang15,Coulson19,Coulson19b} by means of a numerical example.
The paper is concluded in Section~\ref{sec:conclusion}.

\section{Preliminaries}\label{sec:setting}

Let $\mathbb{I}_{[a,b]}$ denote the set of integers in the interval $[a,b]$.
For a vector $x$ and a positive definite matrix $P=P^\top\succ0$, we write $\lVert x\rVert_P=\sqrt{x^\top Px}$.
Further, we denote the minimal and maximal eigenvalue of $P$ by $\lambda_{\min}(P)$ and $\lambda_{\max}(P)$, respectively.
For two matrices $P_1=P_1^\top,P_2=P_2^\top$, we write $\lambda_{\min}(P_1,P_2)=\min\{\lambda_{\min}(P_1),\lambda_{\min}(P_2)\}$, and similarly for $\lambda_{\max}(P_1,P_2)$.
Moreover, $\lVert x\rVert_2$, $\lVert x\rVert_1$, and $\lVert x\rVert_\infty$ denote the Euclidean, $\ell_1$-, and $\ell_\infty$-norm of $x$, respectively.
If the argument is matrix-valued, then we mean the corresponding induced norm.
For $\delta>0$, we define $\mathbb{B}_\delta=\left\{x\in\mathbb{R}^n\mid \lVert x\rVert_2\leq\delta\right\}$.
A sequence $\{x_k\}_{k=0}^{N-1}$ induces the Hankel matrix
\begin{align*}
H_L&(x)\coloneqq\begin{bmatrix}x_0 & x_1& \dots & x_{N-L}\\
x_1 & x_2 & \dots & x_{N-L+1}\\
\vdots & \vdots & \ddots & \vdots\\
x_{L-1} & x_{L} & \dots & x_{N-1}
\end{bmatrix}.
\end{align*}
For a stacked window of the sequence, we write
\begin{align*}
x_{[a,b]}=\begin{bmatrix}x_a\\\vdots\\x_b\end{bmatrix}.
\end{align*}
We denote by $x$ either the sequence itself or the stacked vector $x_{[0,N-1]}$ containing all of its components.
We consider the following standard definition of persistence of excitation.

\begin{definition}
We say that a sequence $\{u_k\}_{k=0}^{N-1}$ with $u_k\in\mathbb{R}^m$ is persistently exciting of order $L$ if $\text{rank}(H_L(u))=mL$.
\end{definition}

Our goal is to control an unknown LTI system, denoted by $G$, of order $n$ with $m$ inputs and $p$ outputs, using only measured input-output data.

\begin{definition}\label{def:trajectory_of}
We say that an input-output sequence $\{u_k,y_k\}_{k=0}^{N-1}$ is a trajectory of an LTI system $G$, if there exists an initial condition $\bar{x}\in\mathbb{R}^n$ as well as a state sequence $\{x_k\}_{k=0}^{N}$ such that
\begin{align*}
x_{k+1}&=Ax_k+Bu_k,\>\>x_0=\bar{x},\\
y_k&=Cx_k+Du_k,
\end{align*}
for $k=0,\dots,N-1$, where $(A,B,C,D)$ is a minimal realization of $G$.
\end{definition}

Note that we define a trajectory of an LTI system as an input-output sequence that can be produced by a minimal realization, entailing controllability and observability of the system.
Extending the results of this paper to systems whose input-output behavior cannot be explained via a minimal realization is an interesting issue for future research.
The following result lays the foundation of the present paper.
It shows that a Hankel matrix, involving a single persistently exciting trajectory, spans the vector space of all system trajectories of an LTI system.
The result originates from behavioral systems theory~\cite{Willems05}, but we employ the formulation in the classical state-space control framework~\cite{Berberich19}.

\begin{theorem}[\cite{Berberich19}]\label{thm:traj_rep}
Suppose $\{u_k^d,y_k^d\}_{k=0}^{N-1}$ is a trajectory of an LTI system $G$, where $u^d$ is persistently exciting of order $L+n$.
Then, $\{\bar{u}_k,\bar{y}_k\}_{k=0}^{L-1}$ is a trajectory of $G$ if and only if there exists $\alpha\in\mathbb{R}^{N-L+1}$ such that
\begin{align}\label{eq:thm_hankel}
\begin{bmatrix}H_L(u^d)\\H_L(y^d)\end{bmatrix}\alpha
=\begin{bmatrix}\bar{u}\\\bar{y}\end{bmatrix}.
\end{align}
\end{theorem}

Recently, Theorem~\ref{thm:traj_rep} has received increasing attention to develop data-driven controllers~\cite{Persis19}, verify dissipativity~\cite{Romer19}, or to design MPC schemes~\cite{Yang15,Coulson19,Coulson19b}.
This is due to the fact that~\eqref{eq:thm_hankel} provides an appealing data-driven characterization of all trajectories of the unknown LTI system, without requiring any prior identification step.
In this paper, we use Theorem~\ref{thm:traj_rep} to develop a data-driven MPC scheme with \emph{provable} %guarantees on stability and constraint satisfaction despite noisy measurements.
stability guarantees despite noisy measurements.
Note that, if a sequence is persistently exciting of order $L$, then it is also persistently exciting of order $\tilde{L}$ for any $\tilde{L}\leq L$.
Therefore, Theorem~\ref{thm:traj_rep} and hence all of our results hold true if $n$ is replaced by a (potentially rough) upper bound.

Although we assume that only input-output data of the unknown system are available, we make extensive use of the fact that an input-output trajectory of length greater than or equal to $n$ induces a unique internal state in some minimal realization of the unknown system.
\begin{comment}
The following result shows that input and initial conditions suffice to fix an $\alpha$, which predicts the corresponding output.
The following result is from data-driven simulation~\cite{Markovsky08,Berberich19} and slightly reformulated.
%
\begin{proposition}\label{prop:DD_sim}
Suppose $\{u_k^d,y_k^d\}_{k=0}^{N-1}$ is a trajectory of an LTI system $G$, where $u$ is persistently exciting of order $L+n$.
Let $\{\bar{u}_k,\bar{y}_k\}_{k=0}^{L-1}$ be an arbitrary trajectory of $G$.
Denote by $\{x_k^d\}_{k=0}^{N-1}$ and $\{\bar{x}_k\}_{k=0}^{L-1}$ the corresponding state trajectories in some minimal realization.
If $\nu\geq n$, then there exists an $\alpha\in\mathbb{R}^{N-L+1}$ such that
%
\begin{align}\label{eq:prop_DD_sim}
\begin{bmatrix}
H_L\left(u^d\right)\\H_1\left(x^d_{[0,N-L]}\right)
\end{bmatrix}\alpha=
\begin{bmatrix}
\bar{u}\\\bar{x}_0
\end{bmatrix}.
\end{align}
%
Further, it holds that $\bar{y}=H_L(y^d)\alpha$.
\end{proposition}
\begin{proof}
Follows directly from~\cite[Proposition 1]{Markovsky08}.
\end{proof}
%
Since the matrix on the left-hand-side of~\eqref{eq:prop_DD_sim} will be essential for many of our arguments, we abbreviate it as 
%
\begin{align}\label{eq:Hux}
H_{ux}\coloneqq\begin{bmatrix}
H_L\left(u^d\right)\\H_1\left(x^d_{[0,N-L]}\right)
\end{bmatrix}.
\end{align}
%
A simple choice to satisfy~\eqref{eq:prop_DD_sim} is $\alpha=H_{ux}^\dagger\begin{bmatrix}\bar{u}\\\bar{x}_0\end{bmatrix}$, where $H_{ux}^\dagger$ denotes the right-inverse of $H_{ux}$.y
\end{comment}
We employ MPC to stabilize a desired equilibrium of the system.
Since a model of this system is not available, we define an equilibrium via input-output pairs.

\begin{definition}\label{def:equil}
We say that an input-output pair $(u^s,y^s)\in\mathbb{R}^{m+p}$ is an equilibrium of an LTI system $G$, if the sequence $\{\bar{u}_k,\bar{y}_k\}_{k=0}^{n}$ with $(\bar{u}_k,\bar{y}_k)=(u^s,y^s)$ for all $k\in\mathbb{I}_{[0,n]}$ is a trajectory of $G$.
\end{definition}

For an equilibrium $(u^s,y^s)$, we define $u^s_n$ and $y^s_n$ as the column vectors containing $n$ times $u^s$ and $y^s$, respectively.
We assume that the system is subject to pointwise-in-time input and output constraints, i.e., $u_t\in\mathbb{U}\subseteq\mathbb{R}^m$, $y_t\in\mathbb{Y}\subseteq\mathbb{R}^p$ for all $t\geq0$, and we assume $(u^s,y^s)\in\text{int}(\mathbb{U}\times\mathbb{Y})$.
Throughout this paper, $\left\{u_k^d,y_k^d\right\}_{k=0}^{N-1}$ denotes an a priori measured data trajectory of length $N$, which is used for prediction as in~\eqref{eq:thm_hankel}.
The predicted input- and output-trajectories at time $t$ over some prediction horizon $L$ are written as $\left\{\bar{u}_k(t),\bar{y}_k(t)\right\}_{k=-n}^{L-1}$.
Note that the time indices start at $k=-n$, since the last $n$ inputs and outputs will be used to invoke a unique initial state at time $t$.
Further, the closed-loop input, the state in some minimal realization, and the output at time $t$ are denoted by $u_t$, $x_t$, and $y_t$, respectively.

\section{Nominal data-driven MPC}\label{sec:tec}
In this section, we propose a simple, nominal data-driven MPC scheme with terminal equality constraints. 
The scheme relies on noise-free measurements to predict future trajectories using Theorem~\ref{thm:traj_rep} and is described in Section~\ref{sec:tec_scheme}.
Under mild assumptions, we prove recursive feasibility, constraint satisfaction, and exponential stability of the closed loop in Section~\ref{sec:stab}.
 
\subsection{Nominal MPC scheme}\label{sec:tec_scheme}
Commonly, MPC relies on a model of the plant to predict future trajectories and to optimize over them.
Theorem~\ref{thm:traj_rep} provides an appealing alternative to a model since~\eqref{eq:thm_hankel} suffices to capture all system trajectories.
Thus, to implement a data-driven MPC scheme, one can simply replace the system dynamics constraint by the constraint that the predicted input-output trajectories satisfy~\eqref{eq:thm_hankel}.
To be more precise, the proposed data-driven MPC scheme minimizes, at time $t$, given the last $n$ input-output pairs, the following open-loop cost\\
\begin{subequations}
\begin{align}
J_L(u_{[t-n,t-1]},y_{[t-n,t-1]},&\alpha(t))=\sum_{k=0}^{L-1}\ell\left(\bar{u}_k(t),\bar{y}_k(t)\right),\\\label{eq:MPC_modela}
\begin{bmatrix}\bar{u}_{[-n,L-1]}(t)\\\bar{y}_{[-n,L-1]}(t)\end{bmatrix}&=\begin{bmatrix}H_{L+n}(u^d)\\H_{L+n}(y^d)\end{bmatrix}\alpha(t),\\
\label{eq:MPC_initial_conditionsa}
\begin{bmatrix}\bar{u}_{[-n,-1]}(t)\\\bar{y}_{[-n,-1]}(t)\end{bmatrix}&=\begin{bmatrix}u_{[t-n,t-1]}\\y_{[t-n,t-1]}\end{bmatrix}.
\end{align}
\end{subequations}
As described above, the constraint~\eqref{eq:MPC_modela} replaces the system dynamics compared to classical model-based MPC schemes.
Further,~\eqref{eq:MPC_initial_conditionsa} ensures that the internal state of the true trajectory aligns with the internal state of the predicted trajectory at time $t$.
Note that the overall length of the trajectory $(\bar{u}(t),\bar{y}(t))$ is $L+n$ since the past $n$ elements $\{\bar{u}_k(t),\bar{y}_k(t)\}_{k=-n}^{-1}$ are used to specify the initial conditions in~\eqref{eq:MPC_initial_conditionsa}.
These initial conditions are specified until time step $t-1$, since the input at time $t$ might already influence the output at time $t$, in case of a feedthrough-element of the plant.
The open-loop cost depends only on the decision variable $\alpha(t)$, since $\bar{u}(t)$ and $\bar{y}(t)$ are fixed implicitly through the dynamic constraint~\eqref{eq:MPC_modela}.
Throughout the paper, we consider quadratic stage costs, which penalize the distance w.r.t. a desired equilibrium $(u^s,y^s)$, i.e.,
\begin{align*}
\ell(\bar{u},\bar{y})=\lVert\bar{u}-u^s\rVert_R^2+\lVert\bar{y}-y^s\rVert_Q^2,
\end{align*}
where $Q,R\succ0$.
In~\cite{Yang15,Coulson19}, it was suggested to directly minimize the above open-loop cost subject to constraints on input and output.
It is well-known that MPC without terminal constraints requires a sufficiently long prediction horizon to ensure stability and constraint satisfaction~\cite{grune2017nonlinear,grune2012nmpc}.
Without such an assumption, the application of MPC can even destabilize an open-loop stable system.
There are two main approaches in the literature to guarantee stability:
a) providing bounds on the minimal required prediction horizon~\cite{grune2012nmpc} and b) including terminal ingredients such as terminal cost functions or terminal region constraints~\cite{Mayne00}.
Both approaches are usually based on model knowledge and thus, it is not straightforward to use them in the present, purely data-driven setting.

In this paper, we consider a simple terminal equality constraint, which can be directly included into the data-driven MPC framework, and which guarantees exponential stability of the closed loop.
To this end, we propose the following data-driven MPC scheme with a terminal equality constraint.
\begin{subequations}\label{eq:term_eq_MPC}
\begin{align}\nonumber
J_L^*(u_{[t-n,t-1]}&,y_{[t-n,t-1]})=\\
\underset{\substack{\alpha(t)\\\bar{u}(t),\bar{y}(t)}}{\min}\>\>&\sum_{k=0}^{L-1}\ell\left(\bar{u}_k(t),\bar{y}_k(t)\right)\\\label{eq:MPC_model}
s.t.&\begin{bmatrix}\bar{u}_{[-n,L-1]}(t)\\\bar{y}_{[-n,L-1]}(t)\end{bmatrix}=\begin{bmatrix}H_{L+n}(u^d)\\H_{L+n}(y^d)\end{bmatrix}\alpha(t),\\
\label{eq:MPC_initial_conditions}
&\begin{bmatrix}\bar{u}_{[-n,-1]}(t)\\\bar{y}_{[-n,-1]}(t)\end{bmatrix}=\begin{bmatrix}u_{[t-n,t-1]}\\y_{[t-n,t-1]}\end{bmatrix},\\\label{eq:term_eq_MPC2}
&\begin{bmatrix}\bar{u}_{[L-n,L-1]}(t)\\\bar{y}_{[L-n,L-1]}(t)\end{bmatrix}=\begin{bmatrix}u^s_n\\y^s_n\end{bmatrix},\\\label{eq:term_eq_MPC3}
&\bar{u}_k(t)\in\mathbb{U},\>\>\bar{y}_k(t)\in\mathbb{Y},\>\>k\in\mathbb{I}_{[0,L-1]}.
\end{align}
\end{subequations}
The terminal equality constraint~\eqref{eq:term_eq_MPC2} implies that $\bar{x}_L(t)$, which is the internal state predicted $L$ steps ahead corresponding to the predicted input-output trajectory, aligns with the steady-state $x^s$ corresponding to $(u^s,y^s)$, i.e., $\bar{x}_L(t)=x^s$ in any minimal realization.
While Problem~\eqref{eq:term_eq_MPC} requires that $(u^s,y^s)$ is an equilibrium of the unknown system in the sense of Definition~\ref{def:equil}, this requirement can be dropped when $(u^s,y^s)$ is replaced by an artificial equilibrium, which is also optimized online (compare~\cite{limon2008mpc}).
The recent paper~\cite{berberich2020tracking} extends the above MPC scheme to such a setting, thereby leading to a significantly larger region of attraction for the closed loop without requiring knowledge of a reachable equilibrium of the unknown system.
As in standard MPC, Problem~\eqref{eq:term_eq_MPC} is solved in a receding horizon fashion, which is summarized in Algorithm~\ref{alg:MPC}.

\begin{algorithm}
\begin{Algorithm}\label{alg:MPC}
\normalfont{\textbf{Data-Driven MPC Scheme}}
\begin{enumerate}
\item At time $t$, take the past $n$ measurements $u_{[t-n,t-1]}$, $y_{[t-n,t-1]}$ and solve~\eqref{eq:term_eq_MPC}.
\item Apply the input $u_t=\bar{u}_0^*(t)$.
\item Set $t=t+1$ and go back to 1).
\end{enumerate}
\end{Algorithm}
\end{algorithm}

With slight abuse of notation, we will denote the open-loop cost and the optimal open-loop cost of~\eqref{eq:term_eq_MPC} by $J_L(x_t,\alpha(t))$ and $J_L^*(x_t)$, respectively, %which are defined to be $J_L(u_{[t-n,t-1]},y_{[t-n,t-1]},\alpha(t))$ and $J_L^*(u_{[t-n,t-1]},y_{[t-n,t-1]})$, respectively, 
where $x_t$ is the state in some minimal realization, induced by $u_{[t-n,t-1]}$, $y_{[t-n,t-1]}$.

\subsection{Closed-loop guarantees}\label{sec:stab}
Without loss of generality, we assume for the analysis that $u^s=0$, $y^s=0$, and thus $x^s=0$.
Further, we define the set of initial states, for which~\eqref{eq:term_eq_MPC} is feasible, by $\mathbb{X}_L=\left\{x\in\mathbb{R}^n\mid J_L^*(x)<\infty\right\}$.
To prove exponential stability of the proposed scheme, we assume that the optimal value function of~\eqref{eq:term_eq_MPC} is quadratically upper bounded.
This is, e.g., satisfied in the present linear-quadratic setting if the constraints are polytopic\footnote{While~\cite{Bemporad02} considered model-based linear-quadratic MPC, the result applies similarly to the present data-driven MPC setting since~\eqref{eq:MPC_model} (together with the initial conditions~\eqref{eq:MPC_initial_conditions}) describes the input-output behavior of the system exactly and thus, both settings are equivalent in the nominal case.}~\cite{Bemporad02}.
%\footnote{Maybe also via the theory from~\cite{rawlings2012postface}, but their proof for general class $\mathcal{K}_\infty$ functions is quite involved.}
%
\begin{assumption}\label{ass:quad_upper_bound}
The optimal value function $J_L^*(x)$ is quadratically upper bounded on $\mathbb{X}_L$, i.e., there exists $c_u>0$ such that $J_L^*(x)\leq c_u\lVert x\rVert_2^2$ for all $x\in\mathbb{X}_L$.
\end{assumption}

Moreover, we assume that the input $u^d$ generating the data used for prediction is sufficiently rich in the following sense.

\begin{assumption}\label{ass:pe}
The input $u^d$ of the data trajectory is persistently exciting of order $L+2n$.
\end{assumption}

Note that we assume persistence of excitation of order $L+2n$, although Theorem~\ref{thm:traj_rep} requires only an order of $L+n$.
This is due to the fact that the reconstructed trajectories in~\eqref{eq:term_eq_MPC} are of length $L+n$ (compared to length $L$ in Theorem~\ref{thm:traj_rep}), since $n$ components are used to fix the initial conditions.
Furthermore, due to the terminal constraints~\eqref{eq:term_eq_MPC2}, the prediction horizon needs to be at least as long as the system order $n$.

\begin{assumption}\label{ass:length_n}
The prediction horizon satisfies $L\geq n$.
\end{assumption}

The following result shows that the MPC scheme based on~\eqref{eq:term_eq_MPC} is recursively feasible, ensures constraint satisfaction, and leads to an exponentially stable closed loop.

\begin{theorem}\label{thm:mpc_tec}
Suppose Assumptions~\ref{ass:quad_upper_bound},~\ref{ass:pe} and~\ref{ass:length_n} are satisfied.
If the MPC problem~\eqref{eq:term_eq_MPC} is feasible at initial time $t=0$, then 
\begin{itemize}
\item[(i)] it is feasible at any $t\in\mathbb{N}$,
\item[(ii)] the closed loop satisfies the constraints, i.e., $u_t\in\mathbb{U}$ and $y_t\in\mathbb{Y}$ for all $t\in\mathbb{N}$,
\item[(iii)] the equilibrium $x^s=0$ is exponentially stable for the resulting closed loop.
\end{itemize}
\end{theorem}
\begin{proof}
Recursive feasibility (i) and constraint satisfaction (ii) follow from standard MPC arguments, i.e., by defining a candidate solution as the shifted, previously optimal solution and appending zero (compare~\cite{Rawlings09}).\\
\textbf{(iii). Exponential Stability}\\
Denote the standard candidate solution mentioned above by $\bar{u}'(t+1),\bar{y}'(t+1),\alpha'(t+1)$.
The cost of this solution is
\begin{align*}
&J_L(x_{t+1},\alpha'(t+1))\\
&=\sum_{k=0}^{L-1}\ell\left(\bar{u}_k'(t+1),\bar{y}_k'(t+1)\right)
=\sum_{k=1}^{L-1}\ell\left(\bar{u}^*_k(t),\bar{y}^*_k(t)\right)\\
&=J_L^*(x_t)-\ell\left(\bar{u}_0^*(t),\bar{y}_0^*(t)\right).
\end{align*}
Hence, it holds that
\begin{align}\label{eq:thm_eq0_proof1}
J_L^*(x_{t+1})\leq J_L^*(x_t)-\ell\left(\bar{u}_0^*(t),\bar{y}_0^*(t)\right).
\end{align}
Since $x$ is the state of an observable (and hence detectable) minimal realization, there exists a matrix $P\succ0$ such that $W(x)=\lVert x\rVert_P^2$ is an input-output-to-state stability (IOSS) Lyapunov function\footnote{Note that, in~\cite[Section 3.2]{cai2008input}, only strictly proper systems with $y=Cx$ are considered, while we allow for more general systems with $y=Cx+Du$.
The result from~\cite{cai2008input} can be extended to $y=Cx+Du$ by considering a modified $\tilde{B}=B+LD$ in~\cite[Inequality (12)]{cai2008input}.}, which satisfies
\begin{align}\label{eq:thm_eq0_proof2}
W(Ax+Bu)-W(x)\leq-\frac{1}{2}\lVert x\rVert_2^2+c_1\lVert u\rVert_2^2 +c_2\lVert y\rVert_2^2,
\end{align}
for all $x\in\mathbb{R}^n,u\in\mathbb{R}^m,y=Cx+Du$, and for suitable $c_1,c_2>0$~\cite{cai2008input}.
Define the candidate Lyapunov function $V(x)=\gamma W(x)+J_L^*(x)$ for some $\gamma>0$.
Note that $V$ is quadratically lower bounded, i.e., $V(x)\geq\gamma W(x)\geq\gamma\lambda_{\min}(P)\lVert x\rVert_2^2$ for all $x\in\mathbb{X}_L$.
Further, $J_L^*$ is quadratically upper bounded by Assumption~\ref{ass:quad_upper_bound}, i.e., $J_L^*(x)\leq c_u\lVert x\rVert_2^2$ for all $x\in\mathbb{X}_L$.
Hence, we have
\begin{align*}
V(x)=J_L^*(x)+\gamma W(x)\leq\left(c_u+\gamma\lambda_{\max}(P)\right)\lVert x\rVert_2^2,
\end{align*}
for all $x\in\mathbb{X}_L$, i.e., $V$ is quadratically upper bounded.
We consider now
\begin{align*}
\gamma=\frac{\lambda_{\min}(Q,R)}{\max\{c_1,c_2\}}>0.
\end{align*}
Along the closed-loop trajectories, using both~\eqref{eq:thm_eq0_proof1} as well as~\eqref{eq:thm_eq0_proof2}, it holds that
\begin{align*}
V(x_{t+1})-V(x_t)\leq&\>\gamma\left(-\frac{1}{2}\lVert x_t\rVert^2_2+c_1\lVert u_t\rVert^2_2 +c_2\lVert y_t\rVert^2_2\right)\\
&-\lVert u_t\rVert_R^2-\lVert y_t\rVert_Q^2\\
\leq&-\frac{\gamma}{2}\lVert x_t\rVert^2_2.
\end{align*}
%
%Thus, $V$ decays exponentially along the closed-loop trajectories, it is quadratically upper and lower bounded, and it satisfies $V(0)=0$.
%Therefore, i
It follows from standard Lyapunov arguments with Lyapunov function $V$ that the equilibrium $x^s=0$ is exponentially stable with region of attraction $\mathbb{X}_L$.
\end{proof}

The proof of Theorem~\ref{thm:mpc_tec} applies standard arguments from model-based MPC with terminal constraints (compare~\cite{Rawlings09}) to the data-driven system description derived in~\cite{Willems05}, similar to the approaches of~\cite{Yang15,Coulson19} which did however not address closed-loop guarantees. 
To handle the fact that the stage cost $\ell$ is merely positive \emph{semi}-definite in the state, detectability of the stage cost is exploited via an IOSS Lyapunov function~\cite{cai2008input}, similar to~\cite{grimm2005model}.
As we will see in Section~\ref{sec:robust}, this analogy between model-based MPC and the proposed data-driven MPC scheme is only present in the nominal case, where the data is noise-free.
For the more realistic case of noisy output measurements, we develop a robust data-driven MPC scheme and we provide a novel theoretical analysis of the closed loop in Section~\ref{sec:robust}, which is the main contribution of this paper.

\begin{remark}\label{rk:complexity_nominal}
We would like to emphasize the simplicity of the proposed MPC scheme.
Without any prior identification step, a single measured data trajectory can be used directly to set up an MPC scheme for a linear system.
Compared to other learning-based MPC approaches such as~\cite{adetola2011robust,Aswani13,tanaskovic2014adaptive,Berkenkamp17,Zanon19}, which require initial model knowledge as well as an online estimation process, the complexity of~\eqref{eq:term_eq_MPC} is similar to classical MPC schemes, which rely on full model knowledge.
To be more precise, the decision variables $\bar{u}(t),\bar{y}(t)$ can be replaced by $\alpha(t)$ via~\eqref{eq:MPC_model} (using a condensed formulation) and hence, since $\alpha(t)\in\mathbb{R}^{N-L-n+1}$, Problem~\eqref{eq:term_eq_MPC} contains in total $N-L-n+1$ decision variables.
For $u^d$ to be persistently exciting of order $L+2n$, it needs to hold that $N-L-2n+1\geq m(L+2n)$.
Assuming equality, Problem~\eqref{eq:term_eq_MPC} hence has $m(L+2n)+n$ free parameters.
On the contrary, a condensed model-based MPC optimization problem contains $mL$ decision variables for the input trajectory (assuming that state measurements are available).
Thus, the online complexity of the proposed data-driven MPC approach is slightly larger ($2mn+n$ additional decision variables) than that of model-based MPC, but it does not require an a priori (offline) identification step.
It is worth noting that the difference in complexity is independent of the horizon $L$.
Moreover, the proposed data-driven MPC is inherently an output-feedback controller since no state measurements are required for its implementation.
Finally, as in model-based MPC, for convex polytopic (or quadratic) constraints $\mathbb{U},\mathbb{Y}$,~\eqref{eq:term_eq_MPC} is a convex (quadratically constrained) quadratic program which can be solved efficiently.
\end{remark}
%!TEX root = ./DD_MPC.tex
\section{Robust data-driven MPC}\label{sec:robust}
In this section, we propose a multi-step robust data-driven MPC scheme and we prove practical exponential stability of the closed loop in the presence of bounded additive output measurement noise.
The scheme includes a slack variable, which is regularized in the cost and compensates noise both in the initial data $(u^d,y^d)$ used for prediction and in the online measurement updates $\left(u_{[t-n,t-1]},y_{[t-n,t-1]}\right)$. 
Section~\ref{sec:scheme} contains the scheme, which is essentially a robust modification of the nominal scheme of Section~\ref{sec:tec}, as well as detailed explanations of the key ingredients.
In Sections~\ref{sec:Lyapunov_bound} and~\ref{sec:prediction_error}, we prove two technical Lemmas, which will be required for our main theoretical results.
Recursive feasibility of the closed loop is proven in Section~\ref{sec:robust_feasibility}.
In Section~\ref{sec:robust_stability}, we show that, under suitable assumptions, the closed loop resulting from the application of the multi-step MPC scheme leads to a practically exponentially stable closed loop.
Moreover, if the noise bound tends to zero, then the region of attraction of the closed loop approaches the set of all initially feasible points.
In this section, we do not consider output constraints, i.e., $\mathbb{Y}=\mathbb{R}^p$.
In~\cite{berberich2020constraints}, we recently extended the results of this section by incorporating tightened output constraints in order to guarantee closed-loop constraint satisfaction despite noisy data.

\subsection{Robust MPC scheme}\label{sec:scheme}

In practice, the output of the unknown LTI system $G$ is usually not available exactly, but might be subject to measurement noise.
This implies that the stacked data-dependent Hankel matrices in~\eqref{eq:thm_hankel} do not span the system's trajectory space exactly and thus, the output trajectories cannot be predicted accurately.
Moreover, noisy output measurements enter the initial conditions in Problem~\eqref{eq:term_eq_MPC}, which deteriorates the prediction accuracy even further.
Therefore, a direct application of the MPC scheme of Section~\ref{sec:tec} may lead to feasibility issues or it may render the closed loop unstable.
In this section, we tackle the issue of noisy measurements with a robust data-driven MPC scheme with terminal constraints.
We consider output measurements with bounded additive noise in the initially available data $\tilde{y}_k^d=y_k^d+\varepsilon_k^d$ as well as in the online measurements $\tilde{y}_k=y_k+\varepsilon_k$.
We make no assumptions on the nature of the noise, but we require that it is bounded as $\lVert\varepsilon_k^d\rVert_{\infty}\leq\bar{\varepsilon}$ and $\lVert\varepsilon_k\rVert_{\infty}\leq\bar{\varepsilon}$ for some $\bar{\varepsilon}>0$. 
Thus, the present setting includes two types of noise.
The data used for the prediction via the Hankel matrices in~\eqref{eq:thm_hankel} is perturbed by $\varepsilon^d$, which can thus be interpreted as a multiplicative model uncertainty.
On the other hand, $\varepsilon$ perturbs the online measurements and hence, the overall control goal is a noisy output-feedback problem.

The key idea to account for noisy measurements is to relax the equality constraint~\eqref{eq:MPC_model}, where the relaxation parameter is penalized appropriately in the cost function.
Given a noisy initial input-output trajectory $\left(u_{[t-n,t-1]},\tilde{y}_{[t-n,t-1]}\right)$ of length $n$, and noisy data $(u^d,\tilde{y}^d)$, we propose the following robust modification of~\eqref{eq:term_eq_MPC}.
\begin{subequations}\label{eq:robust_MPC}
\begin{align}\nonumber
J_L^*&\big(u_{[t-n,t-1]},\tilde{y}_{[t-n,t-1]}\big)=\\\nonumber
\underset{\substack{\alpha(t),\sigma(t)\\\bar{u}(t),\bar{y}(t)}}{\min}&\sum_{k=0}^{L-1}\ell\left(\bar{u}_k(t),\bar{y}_k(t)\right)+\lambda_\alpha\bar{\varepsilon}\lVert\alpha(t)\rVert_2^2+\lambda_\sigma\lVert\sigma(t)\rVert_2^2\\
\label{eq:robust_MPC1} s.t.\>\> &\>\begin{bmatrix}
\bar{u}(t)\\\bar{y}(t)+\sigma(t)\end{bmatrix}=\begin{bmatrix}H_{L+n}\left(u^d\right)\\H_{L+n}\left(\tilde{y}^d\right)\end{bmatrix}\alpha(t),\\\label{eq:robust_MPC2}
&\>\begin{bmatrix}\bar{u}_{[-n,-1]}(t)\\\bar{y}_{[-n,-1]}(t)\end{bmatrix}=\begin{bmatrix}u_{[t-n,t-1]}\\\tilde{y}_{[t-n,t-1]}\end{bmatrix},\\\label{eq:robust_MPC3}
&\>\begin{bmatrix}\bar{u}_{[L-n,L-1]}(t)\\\bar{y}_{[L-n,L-1]}(t)\end{bmatrix}=\begin{bmatrix}u^s_n\\y^s_n\end{bmatrix},\>\>\bar{u}_k(t)\in\mathbb{U},\\
%\label{eq:robust_MPC4}
%&\>\bar{y}_k(t)\in\mathbb{Y}_k,\>\>k\in\mathbb{I}_{[0,L-1]},\\
\label{eq:robust_MPC5}
&\>\lVert\sigma_k(t)\rVert_\infty\leq\bar{\varepsilon}\left(1+\lVert\alpha(t)\rVert_1\right),\>\>k\in\mathbb{I}_{[0,L-1]}.
\end{align}
\end{subequations}
Compared to the nominal MPC problem~\eqref{eq:term_eq_MPC}, the output data trajectory $\tilde{y}^d$ as well as the initial output $\tilde{y}_{[t-n,t-1]}$, which is obtained via online measurements, have been replaced by their noisy counterparts.
Further, the following ingredients have been added:
\begin{itemize}
\item[a)] A slack variable $\sigma$, bounded by~\eqref{eq:robust_MPC5}, to account for the noisy online measurements $\tilde{y}_{[t-n,t-1]}$ and for the noisy data $\tilde{y}^d$ used for prediction, which can be interpreted as a multiplicative model uncertainty,
\item[b)] Quadratic regularization (i.e., \emph{ridge regularization}) of $\alpha$ and $\sigma$ with weights $\lambda_\alpha\bar{\varepsilon},\lambda_\sigma>0$, i.e., the regularization of $\alpha$ depends on the noise level.
\end{itemize}
The above $\ell_2$-norm regularization for $\alpha(t)$ implies that small values of $\lVert\alpha(t)\rVert_2^2$ are preferred.
Since the noisy Hankel matrix $H_{L+n}\left(\tilde{y}^d\right)$ is multiplied by $\alpha(t)$ in~\eqref{eq:robust_MPC1}, this implicitly reduces the influence of the noise on the prediction accuracy.
Intuitively, for increasing $\lambda_\alpha$, the term $\lambda_\alpha\bar{\varepsilon}\lVert\alpha(t)\rVert_2^2$ reduces the ``complexity'' of the data-driven system description~\eqref{eq:robust_MPC1}, similar to regularization methods in linear regression, thus allowing for a tradeoff between tracking performance and the avoidance of overfitting.
The term $\lambda_\sigma\lVert\sigma(t)\rVert_2^2$ yields small values for the slack variable $\sigma(t)$, thus improving the prediction accuracy.
For our theoretical results, $\lambda_\sigma$ can be chosen to be zero since $\sigma(t)$ is already rendered small by the constraint~\eqref{eq:robust_MPC5}.
However, as we discuss in more detail in Remark~\ref{rk:sigma_bound}, the constraint~\eqref{eq:robust_MPC5} is non-convex but can be neglected if $\lambda_\sigma$ is large enough.

An alternative to the present regularization terms are general quadratic regularization kernels, i.e., costs of the form $\lVert\alpha(t)\rVert_{P_\alpha}^2$, $\lVert\sigma(t)\rVert_{P_\sigma}^2$ for suitable matrices $P_\alpha,P_\sigma\succ0$.
Further, in~\cite{Coulson19,Coulson19b}, $\ell_1$-regularizations of $\alpha$ and $\sigma$ were suggested and the resulting MPC scheme, without terminal equality constraints, was successfully applied to a nonlinear stochastic control problem.
However, theoretical guarantees on closed-loop stability were not given.
%Finally, the use of $\ell_1$ stage cost functions $\ell(u,y)$ may be beneficial in certain practical applications... and the analysis is simpler?
Throughout this paper, we consider simple quadratic penalty terms since this simplifies the arguments, but we conjecture that our theoretical results remain to hold for general norms $\lVert\alpha(t)\rVert_p,\lVert\sigma(t)\rVert_q$ with arbitrary $p,q=1,\dots,\infty$.
An interesting open question, which is beyond the scope of this paper, is to investigate the impact of particular choices of regularization norms on the practical performance of the presented MPC approach.
The choice of norms in the constraint~\eqref{eq:robust_MPC5} is independent of the norms in the cost and essentially follows from the $\ell_\infty$-noise bound and the proofs of the value function upper bound (Lemma~\ref{lem:value_fcn_upper_bound}) and recursive feasibility (Proposition~\ref{prop:robust_rec_feas}).

In this section, we study the closed loop resulting from an application of~\eqref{eq:robust_MPC} in an $n$-step MPC scheme (compare~\cite{Gruene15,Worthmann17}).
To be more precise, we consider the scenario that, after solving~\eqref{eq:robust_MPC} online, the first $n$ computed inputs are applied to the system.
Thereafter, the horizon is shifted by $n$ steps, before the whole scheme is repeated (compare Algorithm~\ref{alg:MPC_n_step}).

\begin{algorithm}
\begin{Algorithm}\label{alg:MPC_n_step}
\normalfont{\textbf{$n$-Step Data-Driven MPC Scheme}}
\begin{enumerate}
\item At time $t$, take the past $n$ measurements $u_{[t-n,t-1]}$, $\tilde{y}_{[t-n,t-1]}$ and solve~\eqref{eq:robust_MPC}.
\item Apply the input sequence $u_{[t,t+n-1]}=\bar{u}_{[0,n-1]}^*(t)$ over the next $n$ time steps.
\item Set $t=t+n$ and go back to 1).
\end{enumerate}
\end{Algorithm}
\end{algorithm}

%In model-based MPC, multi-step MPC schemes can have superior theoretical properties.
As we will see in the remainder of this section, for the considered setting with output measurement noise, the multi-step MPC scheme described in Algorithm~\ref{alg:MPC_n_step} has superior theoretical properties compared to its corresponding $1$-step version.
This is mainly due to the terminal equality constraints~\eqref{eq:robust_MPC3}, which complicate the proof of recursive feasibility, similar as in model-based robust MPC with terminal equality constraints and model mismatch.
In particular, we show in this section that, for an $n$-step MPC scheme with a terminal equality constraint, practical exponential stability can be proven.
On the other hand, we comment on the differences for the corresponding $1$-step MPC scheme in Section~\ref{sec:robust_feasibility} (Remark~\ref{rk:one_step}).
In particular, for a $1$-step MPC scheme relying on~\eqref{eq:robust_MPC}, recursive feasibility holds only locally around $(u^s,y^s)$ and thus, only local stability can be guaranteed.
Nevertheless, as we will see in Section~\ref{sec:example} for a numerical example, the practical performance of the $n$-step scheme is almost indistinguishable from the $1$-step scheme.

\begin{remark}
In the nominal case of Section~\ref{sec:tec}, i.e., for $\bar{\varepsilon}=0$,~\eqref{eq:robust_MPC5} implies $\sigma=0$.
Further, the regularization of $\alpha$ vanishes for $\bar{\varepsilon}=0$, and the system dynamics~\eqref{eq:robust_MPC1} as well as the initial conditions~\eqref{eq:robust_MPC2} approach their nominal counterparts.
Thus, for $\bar{\varepsilon}=0$, Problem~\eqref{eq:robust_MPC} reduces to the nominal Problem~\eqref{eq:term_eq_MPC}.
\end{remark}

\begin{remark}\label{rk:sigma_bound}
If the constraint~\eqref{eq:robust_MPC5} is neglected and the input constraint set $\mathbb{U}$ is a convex polytope, then Problem~\eqref{eq:robust_MPC} is a strictly convex quadratic program and can be solved efficiently.
However, the constraint on the slack variable $\sigma$ in~\eqref{eq:robust_MPC5} is non-convex due to the dependence of the right-hand side on $\lVert\alpha(t)\rVert_1$, making it difficult to implement~\eqref{eq:robust_MPC} in an efficient way.
As will become clear later in this section,~\eqref{eq:robust_MPC5} is required to prove recursive feasibility and practical exponential stability.
It may, however, be replaced by the (convex) constraint $\lVert\sigma_k(t)\rVert_\infty\leq c\cdot\bar{\varepsilon}$ for a sufficiently large constant $c>0$, retaining the same theoretical guarantees.
Generally, a larger choice of $c$ increases the region of attraction, but also the size of the exponentially stable set to which the closed loop converges.
Furthermore, the constraint~\eqref{eq:robust_MPC5} can be enforced implicitly by choosing $\lambda_\sigma$ large enough.
In simulation examples, it was observed that the constraint~\eqref{eq:robust_MPC5} is usually satisfied (for suitably large choices of $\lambda_\sigma$) without enforcing it explicitly in the optimization problem and thus, it may in most cases be neglected in the online optimization.
\end{remark}

As in the previous section, we require that the measured input $u^d$ is persistently exciting of order $L+2n$ (Assumption~\ref{ass:pe}).
Further, to establish a local upper bound on the optimal cost of~\eqref{eq:robust_MPC} and to prove recursive feasibility, we require that the horizon $L$ is not shorter than twice the system's order, as captured in the following assumption.

\begin{assumption}\label{ass:length_2n}
The prediction horizon satisfies $L\geq2n$.
\end{assumption}

In some minimal realization, we denote the state trajectory corresponding to $(u^d,y^d)$ by $x^d$.
According to~\cite[Corollary 2]{Willems05}, Assumption~\ref{ass:pe} implies that the matrix
\begin{align}\label{eq:ass_pe_matrix}
H_{ux}=\begin{bmatrix}H_{L+n}\left(u^d\right)\\H_1\left(x^d_{[0,N-L-n]}\right)\end{bmatrix}
\end{align}
has full row rank and thus admits a right-inverse $H_{ux}^\dagger=H_{ux}^\top\left(H_{ux}H_{ux}^\top\right)^{-1}$.
Define the quantity
\begin{align}
c_{pe}\coloneqq\left\lVert H_{ux}^\dagger\right\rVert_2^2.
\end{align}
For our stability results, we will require that $c_{pe}\bar{\varepsilon}$ is bounded from above by a sufficiently small number.
Essentially, this corresponds to a quantitative ``persistence-of-excitation-to-noise''-bound.
To be more precise, abbreviate in the following $U=H_{L+n}(u^d)$ and suppose that
\begin{align}\label{eq:ass_pe_quantitative}
\rho I_{m(L+n)}\preceq UU^\top\preceq\nu I_{m(L+n)}
\end{align}
for scalar constants $\rho,\nu>0$.
Further, define the quantity $c_{pe}^u=\lVert U^\dagger\rVert_2^2=\lVert U^\top(UU^\top)^{-1}\rVert_2^2$.
Then, it holds that
\begin{align}\label{eq:ass_pe_bound}
c_{pe}^u&\leq \left\lVert U^\top\right\rVert_2^2\left\lVert \left(UU^\top\right)^{-1}\right\rVert_2^2\\\nonumber
&=\lambda_{\max}(UU^\top)\cdot\lambda_{\max}\left((UU^\top)^{-1}(UU^\top)^{-1}\right)\\\nonumber
&\leq\frac{\lambda_{\max}(UU^\top)}{\lambda_{\min}(UU^\top)^2} \stackrel{\eqref{eq:ass_pe_quantitative}}{\leq}\frac{\nu}{\rho^2}.
\end{align}
Thus, if a persistently exciting input $u^d$ is multiplied by a constant $c>1$, then $c_{pe}^u$ decreases proportionally to $\frac{1}{c^2}$.
Further, the constant $\rho$ can typically be chosen larger if the data length $N$ increases.
The same arguments can be carried out when assuming a bound of the form~\eqref{eq:ass_pe_quantitative} for the matrix~\eqref{eq:ass_pe_matrix}, but finding a suitable input which generates data achieving such a bound is less obvious.
It is well-known for classical definitions of persistence of excitation that larger excitation of the input implies larger excitation of the state.
Therefore, we conjecture (and we have observed for various practical simulation examples) that $c_{pe}$ decreases with increasing data horizons $N$ and with multiplications of a persistently exciting input data trajectory $u^d$ by a scalar constant greater than one.
This means that, for a given noise level $\bar{\varepsilon}$, robust stability as guaranteed in the following sections can be obtained by choosing a large enough persistently exciting input $u^d$ and/or a sufficiently large data horizon $N$.

Similar to Section~\ref{sec:tec}, we denote the open-loop cost of the robust MPC problem~\eqref{eq:robust_MPC} by $J_L\left(u_{[t-n,t-1]},\tilde{y}_{[t-n,t-1]},\alpha(t),\sigma(t)\right)$, and the optimal cost by $J_L^*\left(u_{[t-n,t-1]},\tilde{y}_{[t-n,t-1]}\right)$.
Moreover, we assume for the analysis that $(u^s,y^s)=(0,0)$.
For the presented robust data-driven MPC scheme, setpoints $(u^s,y^s)\neq (0,0)$ change mainly one quantitative constant in Lemma~\ref{lem:value_fcn_upper_bound}.
We comment on the main differences in the case $(u^s,y^s)\neq(0,0)$ in Section~\ref{sec:robust_feasibility} (Remark~\ref{rk:wlog_zero}).

\subsection{Local upper bound of Lyapunov function}\label{sec:Lyapunov_bound}

In this section, we show that the optimal cost of~\eqref{eq:robust_MPC} admits a quadratic upper bound, similar to the nominal case (cf. Assumption~\ref{ass:quad_upper_bound}).
It is straightforward to see that such an upper bound can not be quadratic in the state $x$ of some minimal realization:
the optimal cost $J_L^*$ depends explicitly on $\alpha^*(t)$ via $\lambda_\alpha\bar{\varepsilon}\lVert\alpha^*(t)\rVert_2^2$, which in turn depends on the past $n$ inputs and outputs $(u_{[t-n,t-1]},y_{[t-n,t-1]})$ through~\eqref{eq:robust_MPC1} and~\eqref{eq:robust_MPC2}. 
Even if the current state is zero, i.e., $x_t=0$, these may in general be arbitrarily large and hence, $\alpha$ and therefore also $J_L^*$ may be arbitrarily large.
Thus, $J_L^*$ does not admit an upper bound in the state $x_t$ of a minimal realization.
To overcome this issue, we consider a different (not minimal) state of the system, defined as
\begin{align*}
\xi_t\coloneqq\begin{bmatrix}u_{[t-n,t-1]}\\y_{[t-n,t-1]}\end{bmatrix}.
\end{align*}
Further, we define the noisy version of $\xi$ as
\begin{align*}
\tilde{\xi}_t\coloneqq\begin{bmatrix}u_{[t-n,t-1]}\\\tilde{y}_{[t-n,t-1]}\end{bmatrix}
=\begin{bmatrix}u_{[t-n,t-1]}\\y_{[t-n,t-1]}+\varepsilon_{[t-n,t-1]}\end{bmatrix}.
\end{align*}
Denote the (not invertible) linear transformation from $\xi$ to an arbitrary but fixed state $x$ in some minimal realization by $T$, i.e., $x_t=T\xi_t$.
Clearly, this implies $\lVert x_t\rVert_2^2\leq\lVert T\rVert_2^2\lVert\xi_t\rVert_2^2\eqqcolon\Gamma_x\lVert\xi_t\rVert_2^2$.
Note that $\xi$ is the state of a detectable state-space realization and thus, there exists an IOSS Lyapunov function $W(\xi)=\lVert \xi\rVert_P^2$, similar to the proof of Theorem~\ref{thm:mpc_tec}.
For some $\gamma>0$, define $V_t\coloneqq J_L^*(\tilde{\xi}_t)+\gamma W(\xi_t)$.
The following result shows that, for the state $\xi$, a meaningful quadratic upper bound on $V$ can be proven.

\begin{lemma}\label{lem:value_fcn_upper_bound}
Suppose Assumptions~\ref{ass:pe} and~\ref{ass:length_2n} hold.
Then, there exists a constant $c_3>0$ as well as a $\delta>0$ such that, for all $\xi_t\in \mathbb{B}_\delta$, Problem~\eqref{eq:robust_MPC} is feasible and $V$ is bounded as
\begin{align}\label{eq:lem:value_fcn_bound}
\gamma\lambda_{\min}(P)\lVert \xi_t\rVert^2_2\leq V_t\leq c_3\lVert \xi_t\rVert_2^2+c_4,
\end{align}
where $c_4=2np\bar{\varepsilon}^2\lambda_\sigma$.
\end{lemma}
\begin{proof}
The lower bound is trivial.
For the upper bound, we construct a feasible candidate solution to Problem~\eqref{eq:robust_MPC} which brings the state $x$ in some minimal realization (and thus the output $y$) to zero in $L$ steps.
Obviously, we have $\bar{u}_{[-n,-1]}(t)=u_{[t-n,t-1]}$ as well as $\bar{y}_{[-n,-1]}(t)=\tilde{y}_{[t-n,t-1]}$ by~\eqref{eq:robust_MPC2}.
By assumption, we have $L\geq2n$ as well as $0\in\text{int}(\mathbb{U})$.
Thus, by controllability, there exists a $\delta>0$ such that for any $x_t$ with $\frac{1}{\Gamma_x}\lVert x_t\rVert_2\leq\lVert\xi_t\rVert_2\leq\delta$, there exists an input trajectory $u_{[t,t+L-1]}\in\mathbb{U}^L$, which brings the state $x_{[t,t+L-1]}$ and the corresponding output $y_{[t,t+L-1]}$ to the origin in $L-n$ steps while satisfying 
%$y_{[t,t+L-1]}\in\text{int}\left(\mathbb{Y}^L\right)$ as well as
%
\begin{align}\label{eq:thm_robust_proof_control}
\left\lVert\begin{bmatrix}u_{[t,t+L-1]}
\\y_{[t,t+L-1]}\end{bmatrix}
\right\rVert_2^2&\leq \Gamma_{uy}\lVert x_t\rVert_2^2
\end{align}
for a suitable constant $\Gamma_{uy}>0$.
As candidate input-output trajectories for~\eqref{eq:robust_MPC}, we choose these $u,y$, i.e., $\bar{u}_{[0,L-1]}(t)=u_{[t,t+L-1]},\bar{y}_{[0,L-1]}(t)=y_{[t,t+L-1]}$.
Moreover, $\alpha(t)$ is chosen as
\begin{align}\label{eq:lem_value_fcn_bound_alpha}
\alpha(t)=H_{ux}^\dagger\begin{bmatrix}u_{[t-n,t+L-1]}\\x_{t-n}
\end{bmatrix},
\end{align}
where $H_{ux}$ is defined in~\eqref{eq:ass_pe_matrix}.
As is described in more detail in~\cite{Markovsky08,Berberich19}, the output of an LTI system is a linear combination of its initial condition and the input, and therefore, the above choice of $\alpha(t)$ implies
\begin{align*}
\begin{bmatrix}H_{L+n}\left(u^d\right)\\H_{L+n}\left(y^d\right)\end{bmatrix}\alpha(t)&=
\begin{bmatrix}\bar{u}_{[-n,L-1]}(t)\\y_{[t-n,t+L-1]}\end{bmatrix}\\
&=
\begin{bmatrix}\bar{u}_{[-n,L-1]}(t)\\
\bar{y}_{[-n,-1]}(t)-\varepsilon_{[t-n,t-1]}
\\\bar{y}_{[0,L-1]}(t)\end{bmatrix},
\end{align*}
where $\varepsilon_{[t-n,t-1]}$ is the true noise instance.
For the slack variable $\sigma$, we choose 
\begin{align}\label{eq:lem_1_sigma}
\begin{split}
\sigma_{[-n,-1]}(t)&=H_n\left(\varepsilon^d_{[0,N-L-1]}\right)\alpha(t)-\varepsilon_{[t-n,t-1]},\\
\sigma_{[0,L-1]}(t)&=H_L\left(\varepsilon^d_{[n,N-1]}\right)\alpha(t),
\end{split}
\end{align}
which implies that~\eqref{eq:robust_MPC1}-\eqref{eq:robust_MPC3} are satisfied.
Finally, writing $e_i$ for a row vector whose $i$-th component is equal to $1$ and which is zero otherwise, we obtain
\begin{align}\label{eq:infty_bound}
\begin{split}
\lVert H_{L+n}(\varepsilon^d)\alpha(t)\rVert_\infty&=\max_{i\in\mathbb{I}_{[1,p(L+n)]}}|e_iH_{L+n}(\varepsilon^d)\alpha(t)|\\
&\leq\bar{\varepsilon}\lVert\alpha(t)\rVert_1.
\end{split}
\end{align}
This implies $\lVert\sigma(t)\rVert_\infty\leq\bar{\varepsilon}\left(\lVert\alpha(t)\rVert_1+1\right)$, which in turn proves that~\eqref{eq:robust_MPC5} is satisfied.

In the following, we employ the above candidate solution to bound the optimal cost and thereby, the function $V$.
Due to observability of the pair $(A,C)$, corresponding to the minimal realization with state $x$, it holds that
\begin{align}\label{eq:lem_value_fcn_bound_ctrb}
\begin{bmatrix}\bar{u}_{[-n,-1]}(t)\\x_{t-n}\end{bmatrix}=
\underbrace{\begin{bmatrix}
I_{mn}&0\\M_1&\Phi_\dagger
\end{bmatrix}}_{M\coloneqq}\xi_t,
\end{align}
where $\Phi_\dagger=(\Phi^\top\Phi)^{-1}\Phi^\top$ is a left-inverse of the observability matrix $\Phi$.
The lower block of~\eqref{eq:lem_value_fcn_bound_ctrb} follows from observability and the linear system dynamics $x_{k+1}=Ax_k+Bu_k,\>y_k=Cx_k+Du_k$ for $k\in\mathbb{I}_{[t-n,t-1]}$, which can be used to compute the matrix $M_1$ depending on $A,B,C,D$.
%Note that $M$ has full rank and $\lVert M\rVert_2>0$.
Hence, $\alpha(t)$ can be bounded as
\begin{align}\nonumber
\lVert&\alpha(t)\rVert_2^2\stackrel{\eqref{eq:lem_value_fcn_bound_alpha}}{\leq}\left\lVert H_{ux}^\dagger\right\rVert_2^2
\left(\left\lVert\bar{u}_{[-n,L-1]}(t)\right\rVert_2^2
+\left\lVert x_{t-n}\right\rVert_2^2\right)\\\nonumber
&=\left\lVert H_{ux}^\dagger\right\rVert_2^2
\left(\left\lVert\bar{u}_{[0,L-1]}(t)\right\rVert_2^2
+\left\lVert\begin{bmatrix}\bar{u}_{[-n,-1]}(t)\\x_{t-n}\end{bmatrix}\right\rVert_2^2\right)\\\label{eq:lem_value_fcn_bound_3}
&\stackrel{\eqref{eq:thm_robust_proof_control},\eqref{eq:lem_value_fcn_bound_ctrb}}{\leq}\underbrace{\left\lVert H_{ux}^\dagger\right\rVert_2^2}_{c_{pe}=}
\left(\Gamma_{uy}\lVert x_t\rVert_2^2+\lVert M\rVert_2^2\lVert\xi_t\rVert_2^2\right).
\end{align}
Using standard norm equivalence properties, it holds for arbitrary $k\in\mathbb{N}$ that
\begin{align}\label{eq:lem_constant}
\left\lVert H_{k}\left(\varepsilon^d_{[0,N-L-n+k-1]}\right)\right\rVert_2^2
\leq c_5k\bar{\varepsilon}^2,
\end{align}
where $c_5\coloneqq p(N-L-n+1)$.
Based on the definition of $\sigma(t)$ in~\eqref{eq:lem_1_sigma}, and using~\eqref{eq:lem_constant} as well as the inequality $(a+b)^2\leq2(a^2+b^2)$, we can bound $\sigma(t)$ in terms of $\alpha(t)$ as
\begin{align}\label{eq:lem_bound_sigma}
\lVert\sigma(t)\rVert_2^2&\leq 2np\bar{\varepsilon}^2+c_5(L+2n)\bar{\varepsilon}^2\lVert\alpha(t)\rVert_2^2.
\end{align}
Combining the above inequalities, $V$ is upper bounded as
\begin{align}
\nonumber
&V_t\leq J_L(\tilde{\xi}_t,\alpha(t),\sigma(t))+\gamma W(\xi_t)\\
\nonumber
&\leq
\lambda_{\max}(Q,R)\Gamma_{uy}\lVert x_t\rVert_2^2+\gamma\lambda_{\max}(P)\lVert \xi_t\rVert_2^2\\
\nonumber
&+\left(\lambda_\alpha+c_5(L+2n)\lambda_\sigma\bar{\varepsilon}\right)c_{pe}\bar{\varepsilon}\left(\Gamma_{uy}\lVert x_t\rVert_2^2+\lVert M\rVert_2^2\lVert\xi_t\rVert_2^2\right)\\\nonumber
&+2np\bar{\varepsilon}^2\lambda_\sigma.
%\label{eq:lemma_proof_upper_bound}
\end{align}
Finally, $x_t$ is bounded by $\xi_t$ as $\lVert x_t\rVert_2^2\leq\Gamma_x\lVert\xi_t\rVert_2^2$, which leads to $V_t\leq c_3\lVert \xi_t\rVert_2^2 +c_4$, where 
\begin{align*}
c_3&=\lambda_{\max}(Q,R)\Gamma_{uy}\Gamma_x+\gamma\lambda_{\max}(P)\\
&+\left(\lambda_\alpha+c_5(L+2n)\lambda_\sigma\bar{\varepsilon}\right)c_{pe}\bar{\varepsilon}\left(\Gamma_{uy}\Gamma_x+\lVert M\rVert_2^2\right),\\
c_4&=2np\bar{\varepsilon}^2\lambda_\sigma.
\end{align*}
\end{proof}
In Section~\ref{sec:tec}, we assumed that the optimal cost is quadratically upper bounded (cf. Assumption~\ref{ass:quad_upper_bound}), which is not restrictive in the nominal linear-quadratic setting.
Lemma~\ref{lem:value_fcn_upper_bound} proves that, under mild assumptions, the optimal cost of the robust MPC problem~\eqref{eq:robust_MPC} admits (locally) a similar upper bound and can thus be seen as the robust counterpart of Assumption~\ref{ass:quad_upper_bound}.

The term $c_4$ is solely due to the slack variable $\sigma$.
This can be explained by noting that, for $\xi_t=0$, $\alpha(t)$, $\bar{u}_{[0,L-1]}(t)$, $\bar{y}_{[0,L-1]}(t)$ can all be chosen to be zero, as long as $\sigma$ compensates the noise, i.e., $\sigma_{[-n,-1]}(t)=-\varepsilon_{[t-n,t-1]}$.

\subsection{Prediction error bound}\label{sec:prediction_error}
Denote the optimizers of~\eqref{eq:robust_MPC} by $\alpha^*(t),\sigma^*(t),\bar{u}^*(t),\bar{y}^*(t)$, and the output trajectory resulting from an open-loop application of $\bar{u}^*(t)$ by $\hat{y}$.
One of the reasons why it is difficult to analyze the presented MPC scheme is the non-trivial relation between the predicted output $\bar{y}^*(t)$ and the ``actual'' output $\hat{y}$.
In the following, we derive a bound on the difference between the two quantities, which will play an important role in proving recursive feasibility %, constraint satisfaction, 
and practical stabiliy of the proposed scheme.
For an integer $k$, define constants $\rho_{2,k},\rho_{\infty,k}$ such that
\begin{align*}
\rho_{2,k}&\geq\left\lVert CA^k\Phi_\dagger\right\rVert_2^2,\\
\rho_{\infty,k}&\geq\left\lVert CA^k\Phi_\dagger\right\rVert_\infty, 
\end{align*}
where $\Phi_\dagger$ is a left-inverse of the observability matrix $\Phi$.

\begin{lemma}\label{lem:prediction_error}
If~\eqref{eq:robust_MPC} is feasible at time $t$, then the following inequalities hold for all $k\in\mathbb{I}_{[0,L-1]}$
\begin{align}
\label{eq:lem_prediction_error1}
\lVert &\hat{y}_{t+k}-\bar{y}^*_k(t)\rVert_2^2\leq8c_5\bar{\varepsilon}^2\lVert\alpha^*(t)\rVert_2^2+2\lVert\sigma_k^*(t)\rVert_2^2\\\nonumber
&+\rho_{2,n+k}\left(16n\bar{\varepsilon}^2\left(c_5\lVert\alpha^*(t)\rVert_2^2+p\right)+4\lVert\sigma_{[-n,-1]}^*(t)\rVert_2^2\right),\\
\label{eq:lem_prediction_error}
\lVert &\hat{y}_{t+k}-\bar{y}^*_k(t)\rVert_\infty\leq\bar{\varepsilon}\lVert\alpha^*(t)\rVert_1+\lVert\sigma_k^*(t)\rVert_\infty \\\nonumber
&+\rho_{\infty,n+k}\left(\bar{\varepsilon}\left(\lVert\alpha^*(t)\rVert_1+1\right)+\left\lVert\sigma_{[-n,-1]}^*(t)\right\rVert_\infty\right),
\end{align}
with $c_5$ from~\eqref{eq:lem_constant}.
\end{lemma}
\begin{proof}
We show only~\eqref{eq:lem_prediction_error} and note that~\eqref{eq:lem_prediction_error1} can be derived following the same steps, using~\eqref{eq:lem_constant} as well as the inequality $(a+b)^2\leq2a^2+2b^2$.
As written above, $\hat{y}$ is the trajectory, resulting from an open-loop application of $\bar{u}^*(t)$ and with initial conditions specified by $\left(u_{[t-n,t-1]},\hat{y}_{[t-n,t-1]}\right)=\left(u_{[t-n,t-1]},y_{[t-n,t-1]}\right)$.
On the other hand, according to~\eqref{eq:robust_MPC1}, $\bar{y}^*(t)$ is comprised as 
\begin{align*}
\bar{y}^*(t)=H_{L+n}\left(\varepsilon^d\right)\alpha^*(t)+H_{L+n}\left(y^d\right)\alpha^*(t)-\sigma^*(t).
\end{align*}
It follows directly from~\eqref{eq:robust_MPC1} and~\eqref{eq:robust_MPC2} that the second term on the right-hand side $H_{L+n}\left(y^d\right)\alpha^*(t)$ is a trajectory of $G$, resulting from an open-loop application of $\bar{u}^*(t)$ and with initial output conditions 
\begin{align*}
\tilde{y}_{[t-n,t-1]}+\sigma_{[-n,-1]}^*(t)-H_n\left(\varepsilon^d_{[0,N-L-1]}\right)\alpha^*(t).
\end{align*}
Define
\begin{align*}
y^-_{[t-n,t+L-1]}=\hat{y}_{[t-n,t+L-1]}-H_{L+n}\left(y^d\right)\alpha^*(t).
\end{align*}
Since $G$ is LTI and $y^-$ contains the difference between two trajectories with the same input, we can assume $\bar{u}^*(t)=0$ for the following arguments without loss of generality.
Hence, $y^-$ is equal to the output component of a trajectory $\left(u^-,y^-\right)$ with zero input and with initial trajectory
\begin{align}\label{eq:lem_2_yminus}
&\begin{bmatrix}u^-_{[t-n,t-1]}\\y^-_{[t-n,t-1]}\end{bmatrix}=\\\nonumber
&\qquad\begin{bmatrix}0\\H_n\left(\varepsilon^d_{[0,N-L-1]}\right)\alpha^*(t)-\varepsilon_{[t-n,t-1]}-\sigma_{[-n,-1]}^*(t)\end{bmatrix}.
\end{align}
The relation to the internal state $x^-$ can be derived as
\begin{align*}
y^-_{[t-n,t-1]}=\Phi x^-_{t-n},
\end{align*}
with the observability matrix $\Phi$.
This leads to the corresponding output at time $t+k$
\begin{align*}
y^-_{t+k}=CA^{n+k}\Phi_\dagger y^-_{[t-n,t-1]},
\end{align*}
where $\Phi_\dagger$ is a left-inverse of $\Phi$.
Using this fact, the expression for $y^-_{[t-n,t-1]}$ in~\eqref{eq:lem_2_yminus}, and the inequality~\eqref{eq:infty_bound}, $\lVert y^-_{t+k}\rVert_\infty$ can be bounded as
\begin{align*}
\lVert y^-_{t+k}\rVert_\infty\leq\rho_{\infty,n+k}\left(\bar{\varepsilon}\left(\lVert\alpha^*(t)\rVert_1+1\right)+\left\lVert\sigma_{[-n,-1]}^*(t)\right\rVert_\infty\right).
\end{align*}
Note that
\begin{align*}
\left\lVert \hat{y}_{t+k}-\bar{y}_k^*(t)\right\rVert_\infty
\leq \lVert y^-_{t+k}\rVert_\infty+\bar{\varepsilon}\lVert\alpha^*(t)\rVert_1+\lVert\sigma^*_k(t)\rVert_\infty,
\end{align*}
which concludes the proof.
\end{proof}

Essentially, Lemma~\ref{lem:prediction_error} gives a bound on the mismatch between the predicted output and the actual output resulting from the open-loop application of $\bar{u}^*(t)$, depending on the optimal solutions $\alpha^*,\sigma^*$, and on system parameters.
In model-based robust MPC schemes, similar bounds are typically used to propagate uncertainty, where the role of the weighting vector $\alpha$ to account for multiplicative uncertainty is replaced by the state $x$ and a model-based uncertainty description (compare~\cite{Kouvaritakis16} for details).
The main difference in the proposed MPC scheme is that the predicted trajectory $\bar{y}^*(t)$ is in general not a trajectory of the system in the sense of Definition~\ref{def:trajectory_of}, corresponding to the input $\bar{u}^*(t)$.
On the contrary, in model-based robust MPC, the predicted trajectory usually satisfies the dynamics of a (nominal) model of the system.

\subsection{Recursive feasibility}\label{sec:robust_feasibility}
The following result shows that, if the proposed robust MPC scheme is feasible at time $t$, then it is also feasible at time $t+n$, %and that the closed-loop output satisfies the constraints
assuming that the noise level is sufficiently small.

\begin{proposition}\label{prop:robust_rec_feas}
Suppose Assumption~\ref{ass:pe} and~\ref{ass:length_2n} hold.
Then, for any $V_{ROA}>0$, there exists an $\bar{\varepsilon}_0>0$ such that for all $\bar{\varepsilon}\leq\bar{\varepsilon}_0$, %if the robust MPC problem~\eqref{eq:robust_MPC} is feasible at initial time $t=0$, then 
if $V_t\leq V_{ROA}$ for some $t\geq 0$, then the optimization problem~\eqref{eq:robust_MPC} is feasible at time $t+n$.
%
\begin{comment}
\begin{itemize}
\item[(i)] the $n$-step MPC scheme is feasible at any $t\in\mathbb{N}$,
\item[(ii)] the closed-loop output with the $n$-step MPC scheme satisfies $y_t\in\mathbb{Y}$ for all $t\in\mathbb{N}$.
\end{itemize}
\end{comment}
\end{proposition}
\begin{proof}
%\textbf{(i). Recursive Feasibility}\\
Suppose the robust MPC problem~\eqref{eq:robust_MPC} is feasible at time $t$ with $V_t\leq V_{ROA}$ and denote the optimizers by $\alpha^*(t),\sigma^*(t),\bar{u}^*(t),\bar{y}^*(t)$.
As in Lemma~\ref{lem:prediction_error}, the trajectory resulting from an open-loop application of $\bar{u}^*(t)$ and with initial conditions specified by $\left(u_{[t-n,t-1]},y_{[t-n,t-1]}\right)$ is denoted by $\hat{y}$.
For $k\in\mathbb{I}_{[-n,L-2n-1]}$, we choose for the candidate input the shifted previously optimal solution, i.e., $\bar{u}_k'(t+n)=\bar{u}_{k+n}^*(t)$.
Over the first $n$ steps, the candidate output must satisfy $\bar{y}_{[-n,-1]}'(t+n)=\tilde{y}_{[t,t+n-1]}$ due to~\eqref{eq:robust_MPC2}.
Further, for $k\in\mathbb{I}_{[0,L-2n-1]}$, the output is chosen as $\bar{y}_k'(t+n)=\hat{y}_{t+n+k}$.
Since $\bar{y}_{[L-n,L-1]}^*(t)=0$ by~\eqref{eq:robust_MPC3}, the prediction error bound of Lemma~\ref{lem:prediction_error} implies that, for any $k\in\mathbb{I}_{[L-n,L-1]}$, it holds that
\begin{align*}
\lVert\hat{y}_{t+k}&\rVert_\infty\leq\bar{\varepsilon}\lVert\alpha^*(t)\rVert_1+\lVert\sigma^*(t)\rVert_\infty\\
&+\rho_{\infty,n+k}\left(\bar{\varepsilon}\left(\lVert\alpha^*(t)\rVert_1+1\right)+\lVert\sigma_{[-n,-1]}^*(t)\rVert_\infty\right).
\end{align*}
For $\bar{\varepsilon}_0$ sufficiently small, $\lVert\sigma^*(t)\rVert_\infty$ becomes arbitrarily small due to~\eqref{eq:robust_MPC5}.
Further, using that $\lambda_\alpha\bar{\varepsilon}\lVert\alpha^*(t)\rVert_2^2\leq J_L^*(u_{t-n,t-1]},\tilde{y}_{[t-n,t-1]})\leq V_{ROA}$, we can bound $\alpha^*(t)$ as
\begin{align*}
\lVert\alpha^*(t)\rVert_1&\leq\sqrt{N-L-n+1}\lVert\alpha^*(t)\rVert_2\\
&\leq\sqrt{N-L-n+1}\sqrt{\frac{V_{ROA}}{\lambda_\alpha\bar{\varepsilon}}}.
\end{align*} 
Hence, if $\bar{\varepsilon}_0$ is sufficiently small, then $\hat{y}_{t+k}$ becomes arbitrarily small at the above time instants.
This implies that the internal state in some minimal realization corresponding to the trajectory $(\bar{u}^*(t),\hat{y})$ at time $t+L-n$, i.e., $\hat{x}_{t+L-n}=\Phi_\dagger\hat{y}_{[t+L-n,t+L-1]}$, approaches zero for $\bar{\varepsilon}\to0$.
Thus, similar to the proof of Lemma~\ref{lem:value_fcn_upper_bound}, there exists an input trajectory $\bar{u}_{[L-2n,L-n-1]}'(t+n)$, which brings the state and the corresponding output $\bar{y}_{[L-2n,L-n-1]}'(t+n)$ to zero in $n$ steps, while satisfying
\begin{align}\label{eq:prop_proof_rec_feas_ctrb}
\left\lVert\begin{bmatrix}\bar{u}_{[L-2n,L-n-1]}'(t+n)
\\\bar{y}_{[L-2n,L-n-1]}'(t+n)\end{bmatrix}
\right\rVert_2^2&\leq \Gamma_{uy}\lVert \hat{x}_{t+L-n}\rVert_2^2.
\end{align}
Moreover, in the interval $\mathbb{I}_{[L-n,L-1]}$, we choose $\bar{u}_{[L-n,L-1]}'(t+n)=0$, $\bar{y}_{[L-n,L-1]}'(t+n)=0$, i.e.,~\eqref{eq:robust_MPC3} is satisfied.
The above arguments imply that
\begin{align*}
\left(\bar{u}'(t+n),\begin{bmatrix}\hat{y}_{[t,t+n-1]}\\\bar{y}_{[0,L-1]}'(t+n)
\end{bmatrix}\right)
\end{align*}
is a trajectory of the unknown LTI system in the sense of Definition~\ref{def:trajectory_of}.
Denote the corresponding internal state in some minimal realization by $\bar{x}'(t+n)$.
We choose $\alpha'(t+n)$ as a corresponding solution to~\eqref{eq:thm_hankel}, i.e., as
\begin{align}\label{eq:prop_proof_rec_feas_1}
\alpha'(t+n)=H_{ux}^\dagger\begin{bmatrix}\bar{u}'_{[-n,L-1]}(t+n)\\x_t\end{bmatrix}
\end{align}
with $H_{ux}$ from~\eqref{eq:ass_pe_matrix}.
Finally, we fix
\begin{align}\label{eq:prop_1_sigma}
\sigma'(t+n)=H_{L+n}\left(\tilde{y}^d\right)\alpha'(t+n)-\bar{y}'(t+n),
\end{align}
which implies that~\eqref{eq:robust_MPC1} holds.
It remains to show that the constraint~\eqref{eq:robust_MPC5} is satisfied.
Over the first $n$ time steps,~\eqref{eq:robust_MPC5} holds since
\begin{align}\nonumber
&\sigma_{[-n,-1]}'(t+n)=H_n\left(\tilde{y}^d_{[0,N-L-1]}\right)\alpha'(t+n)-\tilde{y}_{[t,t+n-1]}\\\nonumber
&\stackrel{\eqref{eq:prop_proof_rec_feas_1}}{=}H_n\left(\varepsilon^d_{[0,N-L-1]}\right)\alpha'(t+n)+y_{[t,t+n-1]}-\tilde{y}_{[t,t+n-1]}\\\label{eq:prop_proof_rec_feas_sigma1}
&=H_n\left(\varepsilon^d_{[0,N-L-1]}\right)\alpha'(t+n)-\varepsilon_{[t,t+n-1]}.
\end{align}
Further, using the definition of $\sigma'(t+n)$ in~\eqref{eq:prop_1_sigma} and the bound~\eqref{eq:infty_bound}, we obtain
\begin{align}\label{eq:prop_proof_rec_feas_sigma2}
\lVert&\sigma_{[0,L-1]}'(t+n)\rVert_\infty\leq\bar{\varepsilon}\lVert\alpha'(t+n)\rVert_1\\\nonumber
&+\underbrace{\left\lVert H_{L}\left(y^d_{[n,N-1]}\right)\alpha'(t+n)-\bar{y}_{[0,L-1]}'(t+n)\right\rVert_\infty}_{=0},
\end{align}
and thus,~\eqref{eq:robust_MPC5} holds.
\end{proof}

Proposition~\ref{prop:robust_rec_feas} shows that, for any sublevel set of the Lyapunov function $V$, there exists a sufficiently small noise bound $\bar{\varepsilon}_0$ such that, for any $\bar{\varepsilon}\leq\bar{\varepsilon}_0$ and any state starting in the sublevel set at time $t$, the $n$-step MPC scheme is feasible at time $t+n$.
In particular, the required noise bound decreases if the size of the sublevel set, i.e., $V_{ROA}$, increases and vice versa.
This can be explained by noting that the noise in~\eqref{eq:robust_MPC1} corresponds to a multiplicative uncertainty, which affects the prediction accuracy more strongly if the current state is further away from the origin and hence the Lyapunov function $V_t$ is larger.
We note that this does not imply recursive feasibility of the $n$-step MPC scheme in the standard sense since it remains to be shown that the sublevel set $V_t\leq V_{ROA}$ is invariant, which will be proven in Section~\ref{sec:robust_stability}.
In our main result, the set of initial states for which $V_0\leq V_{ROA}$ will play the role of the guaranteed region of attraction of the closed-loop system.

The input candidate solution used to prove recursive feasibility in Proposition~\ref{prop:robust_rec_feas} is analogous to a candidate solution one would use to show robust recursive feasibility in model-based robust MPC with terminal equality constraints.
The output candidate solution is sketched in Figure~\ref{fig:sketch_candidate}.
Up to time $L-2n-1$, $\bar{y}'(t+n)$ is equal to $\hat{y}$ (shifted by $n$ times steps), which is the output, resulting from an open-loop application of $\bar{u}^*(t)$.
This choice together with the prediction error bound of Lemma~\ref{lem:prediction_error} implies that the internal state corresponding to $\bar{y}'(t+n)$ at time $L-2n$ is close to zero.
Thus, by controllability, there exists an input trajectory satisfying the input constraints, which brings the state and the output to zero in $n$ steps.
In the interval $\mathbb{I}_{[L-2n,L-n-1]}$, the candidate output is chosen as this trajectory.
This also implies that the choice $\bar{y}_{[L-n,L-1]}'(t+n)=0$ makes the candidate solution between $0$ and $L-1$, i.e., $\left(\bar{u}_{[0,L-1]}'(t+n),\bar{y}_{[0,L-1]}'(t+n)\right)$, a trajectory\footnote{
In most practical cases, $(\bar{u}^*(t),\bar{y}^*(t))$ are not trajectories of the system due to the slack variable $\sigma$ and the noise.} of the unknown system $G$ in the sense of Definition~\ref{def:trajectory_of}.
Finally, the suggested candidate input is also similar to~\cite{Yu14}, where inherent robustness of quasi-infinite horizon (model-based) MPC is shown.

\begin{figure}
		\includegraphics[width=0.49\textwidth]{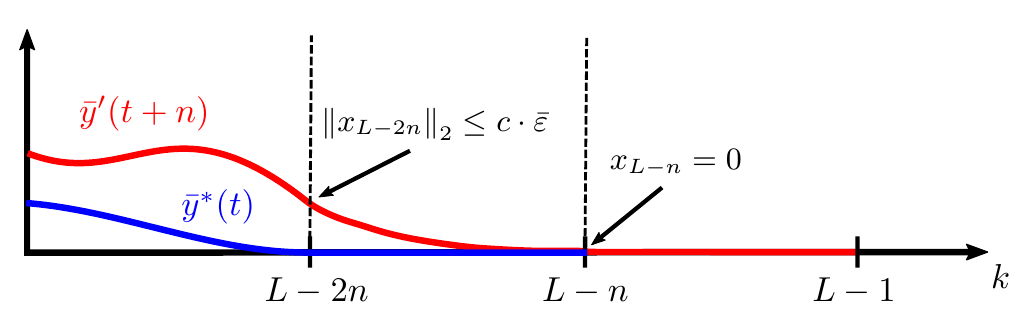}
		\caption{Sketch of the candidate output for recursive feasibility.
		Due to the terminal equality constraints~\eqref{eq:robust_MPC3}, the last $n$ steps of the optimal predicted output $\bar{y}^*(t)$ are equal to zero.
		According to the prediction error bound derived in Lemma~\ref{lem:prediction_error}, this implies that the state resulting from an open-loop application of the optimal input $\bar{u}^*(t)$ is small at time $L-2n$, provided that $\bar{\varepsilon}$ is sufficiently small.
		Therefore, a candidate solution $\bar{y}'(t+n)$ can be constructed by appending the open-loop output $\hat{y}$ by a local deadbeat controller, which steers the state to the origin in $n$ steps.
		}	\label{fig:sketch_candidate}
\end{figure}

\begin{remark}\label{rk:one_step}
For a $1$-step MPC scheme, a similar argument to prove recursive feasibility can be applied, given that $\bar{u}^*_{[L-2n,L-n-1]}(t)$ and $\bar{y}^*_{[L-2n,L-n-1]}(t)$ (and hence $\hat{y}_{[t+L-2n,t+L-n-1]}$) are close to zero.
This is required to construct a feasible input which steers the state and the corresponding output to zero, similar to the proof of Proposition~\ref{prop:robust_rec_feas}, and it is, e.g., the case if the initial state $x_t$ is close to zero.
That is, the result of Proposition~\ref{prop:robust_rec_feas} holds locally for a $1$-step MPC scheme, as expected based on model-based MPC with terminal equality constraints under disturbances using inherent robustness properties.
\end{remark}

\begin{remark}
\label{rk:wlog_zero}
As mentioned in Section~\ref{sec:scheme}, all of our theoretical guarantees for the presented robust MPC scheme can be straightforwardly extended to the case $(u^s,y^s)\neq0$, with the corresponding steady-state $\xi^s\neq0$.
The main difference lies in the bound~\eqref{eq:lem:value_fcn_bound}, which becomes $V_t\leq \tilde{c}_3\lVert \xi_t-\xi^s\rVert_2^2+\tilde{c}_4$ for constants $\tilde{c}_3\neq c_3,\tilde{c}_4\neq c_4$, where $\tilde{c}_3$ can be made arbitrarily close to $c_3$.
On the other hand, $\tilde{c}_4$ changes depending on $\xi^s$, since the right-hand side of~\eqref{eq:lem_value_fcn_bound_3} would need to be proportional to $\lVert \xi_t-\xi^s\Vert_2^2+\lVert \xi^s\rVert_2^2$.
The same phenomenon can be observed in a bound of $\alpha'(t+n)$ based on~\eqref{eq:prop_proof_rec_feas_1}, which will be used in the stability proof.
As will become clear later in this section, such changes in the bound of $\alpha'(t+n)$ as well as in the constant $\tilde{c}_4$ do not affect our qualitative theoretical results, but they may potentially (quantitatively) deterioriate the robustness w.r.t. the noise level $\bar{\varepsilon}$.
Intuitively, this can be explained by noting that~\eqref{eq:robust_MPC1} corresponds to a multiplicative uncertainty and thus, stabilization of the origin is simpler than stabilization of any other equilibrium.
Since equilibria with $(u^s,y^s)\neq0$ require a significantly more involved notation, we omit this extension.
\end{remark}

\subsection{Practical exponential stability}\label{sec:robust_stability}

The following is our main stability result.
It shows that, under Assumptions~\ref{ass:pe} and~\ref{ass:length_2n}, for a low noise amplitude and large persistence of excitation, and for suitable regularization parameters, the application of the scheme~\eqref{eq:robust_MPC} as described in Algorithm~\ref{alg:MPC_n_step} leads to a practically exponentially stable closed loop.
\begin{theorem}\label{thm:robust}
Suppose Assumptions~\ref{ass:pe} and~\ref{ass:length_2n} hold.
Then, for any $V_{ROA}>0$, there exist constants $\underline{\lambda}_\alpha,\overline{\lambda}_\alpha,\underline{\lambda}_\sigma,\overline{\lambda}_\sigma>0$ such that, for all $\lambda_\alpha,\lambda_\sigma$ satisfying
\begin{align}\label{eq:thm_bounds1}
\begin{split}
\underline{\lambda}_\alpha\leq\lambda_\alpha\leq\overline{\lambda}_\alpha,\quad\underline{\lambda}_\sigma\leq\lambda_\sigma\leq\overline{\lambda}_\sigma,
\end{split}
\end{align}
there exist constants $\bar{\varepsilon}_0,\bar{c}_{pe}>0$, as well as a continuous, strictly increasing $\beta:[0,\bar{\varepsilon}_0]\to[0,V_{ROA}]$ with $\beta(0)=0$, such that, for all $\bar{\varepsilon},c_{pe}$ satisfying
\begin{align}\label{eq:thm_bounds2}
\bar{\varepsilon}\leq\bar{\varepsilon}_0,\quad c_{pe}\bar{\varepsilon}\leq\overline{c}_{pe},
\end{align}
the sublevel set $V_t\leq V_{ROA}$ is invariant and $V_t$ converges exponentially to $V_t\leq\beta(\bar{\varepsilon})$ in closed loop with the $n$-step MPC scheme for all initial conditions for which $V_0\leq V_{ROA}$.
\end{theorem}
\begin{proof}
The proof consists of three parts:
First, we bound the increase in the Lyapunov function $V$.
Thereafter, we prove that, for suitably chosen bounds on the parameters, there exists a function $\beta$, which satisfies the above requirements.
Finally, we show invariance of the sublevel set $V_t\leq V_{ROA}$ and exponential convergence of $V_t$ to $V_t\leq\beta(\bar{\varepsilon})$.\\
\textbf{(i). Practical Stability}\\
Suppose Problem~\eqref{eq:robust_MPC} is feasible at time $t$ and let $V_{ROA}>0$ be arbitrary.
Further, let $\bar{\varepsilon}_0$ be sufficiently small such that Proposition~\ref{prop:robust_rec_feas} is applicable.
The cost of the candidate solution derived in Proposition~\ref{prop:robust_rec_feas} at time $t+n$ is
\begin{align*}
J_L&\left(u_{[t,t+n-1]},\tilde{y}_{[t,t+n-1]},\alpha'(t+n),\sigma'(t+n)\right)\\
=&\sum_{k=0}^{L-1}\ell\left(\bar{u}_k'(t+n),\bar{y}_k'(t+n)\right)+\lambda_\alpha\bar{\varepsilon}\lVert\alpha'(t+n)\rVert_2^2\\
&+\lambda_\sigma\lVert\sigma'(t+n)\rVert_2^2.
\end{align*}
Thus, we obtain for the optimal cost
\begin{align}\nonumber
&J_L^*(u_{[t,t+n-1]},\tilde{y}_{[t,t+n-1]})\\
\nonumber
&\leq J_L\left(u_{[t,t+n-1]},\tilde{y}_{[t,t+n-1]},\alpha'(t+n),\sigma'(t+n)\right)\\
\label{eq:thm_proof_value_fcn_diff}
&= J_L^*(u_{[t-n,t-1]},\tilde{y}_{[t-n,t-1]})-\sum_{k=0}^{L-1}\ell\left(\bar{u}_k^*(t),\bar{y}_k^*(t)\right)\\
\nonumber
&-\lambda_\alpha\bar{\varepsilon}\lVert\alpha^*(t)\rVert_2^2-\lambda_\sigma\lVert\sigma^*(t)\rVert_2^2+\lambda_\alpha\bar{\varepsilon}\lVert\alpha'(t+n)\rVert_2^2\\
\nonumber
&+\lambda_\sigma\lVert\sigma'(t+n)\rVert_2^2+\sum_{k=0}^{L-1}\ell(\bar{u}_k'(t+n),\bar{y}_k'(t+n)).
\end{align}
In the following key technical part of the proof (Parts (i.i)-(i.iv)), we derive useful bounds for most terms on the right-hand side of~\eqref{eq:thm_proof_value_fcn_diff}.
This will lead to a decay bound of the optimal cost which is then used to prove practical exponential stability of the closed loop.\\
\textbf{(i.i) Stage Cost Bounds}\\
We first bound those terms in~\eqref{eq:thm_proof_value_fcn_diff}, which involve the stage cost.
The above difference can be decomposed as
\begin{align}\label{eq:thm_proof_stage_cost_diff}
&\sum_{k=0}^{L-1}\ell(\bar{u}_k'(t+n),\bar{y}_k'(t+n))-\sum_{k=0}^{L-1}\ell\left(\bar{u}_k^*(t),\bar{y}_k^*(t)\right)\\\nonumber
&=\sum_{k=L-2n}^{L-n-1}\ell\left(\bar{u}_k'(t+n),\bar{y}_k'(t+n)\right)-\sum_{k=0}^{n-1}\ell\left(\bar{u}_k^*(t),\bar{y}_k^*(t)\right)\\\nonumber
&+\sum_{k=0}^{L-2n-1}\ell\left(\bar{u}_k'(t+n),\bar{y}_k'(t+n)\right)-\ell\left(\bar{u}_{k+n}^*(t),\bar{y}_{k+n}^*(t)\right),
\end{align}
where we use that $\bar{u}_k'(t+n),\bar{y}_k'(t+n),\bar{u}_k^*(t),\bar{y}_k^*(t)$ are all zero for $k\in\mathbb{I}_{[L-n,L-1]}$ due to~\eqref{eq:robust_MPC3}.
To bound the first term on the right-hand side of~\eqref{eq:thm_proof_stage_cost_diff}, note that
\begin{align*}
\lVert \hat{x}_{t+L-n}\rVert_2^2\leq\lVert\Phi_\dagger\rVert_2^2\lVert\hat{y}_{[t+L-n,t+L-1]}\rVert_2^2,
\end{align*}
with $\hat{x}_{t+L-n}$ as in the proof of Proposition~\ref{prop:robust_rec_feas}.
Further, since $\bar{y}_{[L-n,L-1]}^*(t)=0$, $\hat{y}$ can be bounded in the considered time interval as in~\eqref{eq:lem_prediction_error1}, i.e.,
\begin{align*}
&\lVert\hat{y}_{[t+L-n,t+L-1]}\rVert_2^2\leq8c_5n\bar{\varepsilon}^2\lVert\alpha^*(t)\rVert_2^2+2\lVert\sigma^*(t)\rVert_2^2\\
&+\sum_{k=L-n}^{L-1}\rho_{2,n+k}\\
&\cdot\left(16n\bar{\varepsilon}^2\left(c_5\lVert\alpha^*(t)\rVert_2^2+p\right)+4\lVert\sigma^*_{[-n,-1]}(t)\rVert_2^2\right).
\end{align*}
Hence, it holds that
\begin{align}\label{eq:thm_proof_stage_cost_bound_1}
&\sum_{k=L-2n}^{L-n-1}\ell\left(\bar{u}_k'(t+n),\bar{y}_k'(t+n)\right)\\
\nonumber
&\stackrel{\eqref{eq:prop_proof_rec_feas_ctrb}}{\leq}\lambda_{\max}(Q,R)\Gamma_{uy}\lVert\hat{x}_{t+L-n}\rVert_2^2\\
\nonumber
&\leq\lambda_{\max}(Q,R)\Gamma_{uy}\lVert\Phi_\dagger\rVert_2^2\Big(8c_5n\bar{\varepsilon}^2\lVert\alpha^*(t)\rVert_2^2+2\lVert\sigma^*(t)\rVert_2^2\\
\nonumber
&+\sum_{k=L-n}^{L-1}\rho_{2,n+k}\\
\nonumber
&\cdot\left(16n\bar{\varepsilon}^2\left(c_5\lVert\alpha^*(t)\rVert_2^2+p\right)+4\lVert\sigma^*_{[-n,-1]}(t)\rVert_2^2\right)\Big).
\end{align}
Next, we bound the difference between the third and the fourth term on the right-hand side of~\eqref{eq:thm_proof_stage_cost_diff}.
The following relations are readily derived:
\begin{align}\nonumber
&\lVert\bar{y}_k'(t+n)\rVert_Q^2-\lVert\bar{y}_{k+n}^*(t)\rVert_Q^2\\\nonumber
=&\lVert\bar{y}_k'(t+n)-\bar{y}^*_{k+n}(t)+\bar{y}_{k+n}^*(t)\rVert_Q^2-\lVert\bar{y}_{k+n}^*(t)\rVert_Q^2\\\nonumber
=&\lVert\bar{y}_k'(t+n)-\bar{y}_{k+n}^*(t)\rVert_Q^2\\\label{eq:robust_thm_proof_bound5}
&+2\left(\bar{y}_k'(t+n)-\bar{y}_{k+n}^*(t)\right)^\top Q\bar{y}_{k+n}^*(t)\\\nonumber
\leq&\lVert\bar{y}_k'(t+n)-\bar{y}_{k+n}^*(t)\rVert_Q^2\\\nonumber
&+2\lVert\bar{y}_k'(t+n)-\bar{y}_{k+n}^*(t)\rVert_Q\lVert\bar{y}_{k+n}^*(t)\rVert_Q.
\end{align}
By using $2\lVert\bar{y}_{k+n}^*(t)\rVert_Q\leq1+\lVert\bar{y}_{k+n}^*(t)\rVert_Q^2$ as well as $\lVert\bar{y}_{k+n}^*(t)\rVert_Q^2\leq V_{ROA}$, we arrive at
\begin{align}\label{eq:thm_proof_auxiliary_bound1}
2\lVert\bar{y}_k'&(t+n)-\bar{y}_{k+n}^*(t)\rVert_Q\lVert\bar{y}_{k+n}^*(t)\rVert_Q\\
\nonumber
\leq&\lVert\bar{y}_k'(t+n)-\bar{y}_{k+n}^*(t)\rVert_Q\left(1+V_{ROA}\right).
\end{align}
Therefore, since the inputs coincide over the considered time interval, and due to~\eqref{eq:robust_thm_proof_bound5} as well as~\eqref{eq:thm_proof_auxiliary_bound1}, it holds that
\begin{align}\nonumber
&\sum_{k=0}^{L-2n-1}\ell\left(\bar{u}_k'(t+n),\bar{y}_k'(t+n)\right)-\ell\left(\bar{u}_{k+n}^*(t),\bar{y}_{k+n}^*(t)\right)\\
\label{eq:thm_proof_stage_cost_bound_2}
&\leq\sum_{k=0}^{L-2n-1}\lVert\bar{y}_k'(t+n)-\bar{y}_{k+n}^*(t)\rVert_Q^2\\\nonumber
&\quad+\lVert\bar{y}_k'(t+n)-\bar{y}_{k+n}^*(t)\rVert_Q\left(1+V_{ROA}\right).
\end{align}
The difference $\lVert\bar{y}_k'(t+n)-\bar{y}_{k+n}^*(t)\rVert_Q$ can be bounded similar to Lemma~\ref{lem:prediction_error}.
Using the constraint~\eqref{eq:robust_MPC5} to bound $\lVert\sigma^*(t)\rVert_2$, it can be shown that the bound is of the form $\lVert\bar{y}_k'(t+n)-\bar{y}_{k+n}^*(t)\rVert_Q\leq \tilde{C}_1\lVert\alpha^*(t)\rVert_2+\tilde{C}_2\leq\tilde{C}_1\left(1+\lVert\alpha^*(t)\rVert_2^2\right)+\tilde{C}_2$, where both $\tilde{C}_1$ and $\tilde{C}_2$ are proportional to $\bar{\varepsilon}$.
Hence, applying Lemma~\ref{lem:prediction_error} to~\eqref{eq:thm_proof_stage_cost_bound_2}, the sum of~\eqref{eq:thm_proof_stage_cost_bound_1} and~\eqref{eq:thm_proof_stage_cost_bound_2} can be bounded as $C_1\lVert\alpha^*(t)\rVert_2^2+C_2\lVert\sigma^*(t)\rVert_2^2+C_3$ for suitable $C_i>0$, where $C_1$ and $C_3$ are quadratic in $\bar{\varepsilon}$ and vanish for $\bar{\varepsilon}=0$.
Therefore, if $\underline{\lambda}_\alpha$ and $\underline{\lambda}_\sigma$ are sufficiently large, then~\eqref{eq:thm_proof_value_fcn_diff} implies
\begin{align}
\nonumber
&J_L^*(u_{[t,t+n-1]},\tilde{y}_{[t,t+n-1]})\\
\label{eq:thm_proof_value_fcn_diff2}
&\leq J_L^*(u_{[t-n,t-1]},\tilde{y}_{[t-n,t-1]})-\sum_{k=0}^{n-1}\ell\left(\bar{u}_k^*(t),\bar{y}_k^*(t)\right)\\
\nonumber
&+\lambda_\alpha\bar{\varepsilon}\lVert\alpha'(t+n)\rVert_2^2+\lambda_\sigma\lVert\sigma'(t+n)\rVert_2^2
+c_6,
\end{align}
for a suitable constant $c_6>0$, which is quadratic in $\bar{\varepsilon}$ and vanishes for $\bar{\varepsilon}=0$.
\\
\textbf{(i.ii) Bound of $\mathbf{\lVert\sigma'(t+n)\rVert_2^2}$}\\
By applying standard norm bounds to the slack variable candidate $\sigma'(t+n)$ as defined in~\eqref{eq:prop_1_sigma} (compare also~\eqref{eq:prop_proof_rec_feas_sigma1} and~\eqref{eq:prop_proof_rec_feas_sigma2}), we obtain
\begin{align}\label{eq:thm_proof_sigma_bound}
\lVert\sigma'(t+n)\rVert_2^2&\leq 2np\bar{\varepsilon}^2+c_5(L+2n)\bar{\varepsilon}^2\lVert\alpha'(t+n)\rVert_2^2,
\end{align}
with $c_5=p(N-L-n+1)$ as in~\eqref{eq:lem_constant}.\\
\textbf{(i.iii) Bound of $\mathbf{\lVert\alpha'(t+n)\rVert_2^2}$}\\
For the weighting vector $\alpha'(t+n)$, it holds that
\begin{align*}
&\lVert\alpha'(t+n)\rVert_2^2\stackrel{\eqref{eq:prop_proof_rec_feas_1}}{\leq} c_{pe}\left\lVert\begin{bmatrix}\bar{u}'_{[-n,L-1]}(t+n)\\\bar{x}_{-n}'(t+n)\end{bmatrix}\right\rVert_2^2\\
&= c_{pe}\left( \lVert x_t\rVert_2^2+\lVert\bar{u}_{[-n,L-1]}'(t+n)\rVert_2^2\right)=c_{pe}\lVert x_t\rVert_2^2\\
&+c_{pe}\left(\lVert\bar{u}_{[0,L-n-1]}^*(t)\rVert_2^2+\lVert\bar{u}'_{[L-2n,L-n-1]}(t+n)\rVert_2^2\right).
\end{align*}
Similar to~\eqref{eq:thm_proof_stage_cost_bound_1}, we can use~\eqref{eq:prop_proof_rec_feas_ctrb} to bound the last term as
\begin{align}\label{eq:thm_proof_alpha_bound}
&\lVert\bar{u}'_{[L-2n,L-n-1]}(t+n)\rVert_2^2\\\nonumber
&\leq\Gamma_{uy}\lVert\Phi_\dagger\rVert_2^2\Big(8c_5n\bar{\varepsilon}^2\lVert\alpha^*(t)\rVert_2^2+2\lVert\sigma^*(t)\rVert_2^2+\sum_{k=L-n}^{L-1}\rho_{2,n+k}\\
\nonumber
&\cdot\left(16n\bar{\varepsilon}^2\left(c_5\lVert\alpha^*(t)\rVert_2^2+p\right)+4\lVert\sigma^*_{[-n,-1]}(t)\rVert_2^2\right)\Big).
\end{align}
The bound~\eqref{eq:thm_proof_alpha_bound} is of the same form as~\eqref{eq:thm_proof_stage_cost_bound_1} and~\eqref{eq:thm_proof_stage_cost_bound_2}.
Due to this fact, using the bound~\eqref{eq:thm_proof_sigma_bound} for $\sigma'(t+n)$, and by potentially choosing $\underline{\lambda}_\alpha$ and $\underline{\lambda}_\sigma$ larger,~\eqref{eq:thm_proof_value_fcn_diff} implies
\begin{align}
\nonumber
&J_L^*(u_{[t,t+n-1]},\tilde{y}_{[t,t+n-1]})\\
\label{eq:thm_proof_value_fcn_diff3}
&\leq J_L^*(u_{[t-n,t-1]},\tilde{y}_{[t-n,t-1]})-\sum_{k=0}^{n-1}\ell\left(\bar{u}_k^*(t),\bar{y}_k^*(t)\right)\\
\nonumber
&+\left(\lambda_\alpha+\lambda_\sigma c_7\right)\left(\lVert x_t\rVert_2^2+\lVert\bar{u}_{[0,L-n-1]}^*(t)\rVert_2^2\right)c_{pe}\bar{\varepsilon}+c_8,
\end{align}
for suitable constants $c_7,c_8>0$, which vanish for $\bar{\varepsilon}=0$.\\
\textbf{(i.iv) IOSS Bound}\\
As in the proof of Theorem~\ref{thm:mpc_tec}, we consider now $V_t=J_L^*(\tilde{\xi}_t)+\gamma W(\xi_t)$ with the IOSS Lyapunov function $W$ for some $\gamma>0$.
It follows directly from~\eqref{eq:thm_eq0_proof2},~\eqref{eq:thm_proof_value_fcn_diff3}, and from $\lVert x_t\rVert_2^2\leq\Gamma_x\lVert \xi_t\rVert_2^2$ that
\begin{align}\label{eq:thm_proof_value_fcn_diff4}
&V_{t+n}-V_t\leq-\sum_{k=0}^{n-1}\ell\left(\bar{u}_k^*(t),\bar{y}_k^*(t)\right)\\
\nonumber
&+\gamma(-\frac{1}{2}\lVert \xi_{[t,t+n-1]}\rVert_2^2+c_1\lVert u_{[t,t+n-1]}\rVert_2^2+c_2\lVert y_{[t,t+n-1]}\rVert_2^2)\\
\nonumber
&+\left(\lambda_\alpha+\lambda_\sigma c_7\right)\left(\Gamma_x\lVert \xi_t\rVert_2^2+\lVert\bar{u}_{[0,L-n-1]}^*(t)\rVert_2^2\right)c_{pe}\bar{\varepsilon}+c_8.
\end{align}
The identity $(a+b)^2\leq2(a^2+b^2)$ yields
\begin{align*}
\lVert y_{[t,t+n-1]}\rVert_2^2&\leq2\lVert\bar{y}_{[0,n-1]}^*(t)\rVert_2^2\\
&+2\lVert y_{[t,t+n-1]}-\bar{y}_{[0,n-1]}^*(t)\rVert_2^2,
\end{align*}
where the latter term can again be bounded using Lemma~\ref{lem:prediction_error}.
Similar to the earlier steps of this proof, the components of the bound $\lVert y_{[t,t+n-1]}-\bar{y}_{[0,n-1]}^*(t)\rVert_2^2$ vanish in~\eqref{eq:thm_proof_value_fcn_diff4} if $\underline{\lambda}_\sigma,\underline{\lambda}_\alpha$ are chosen sufficiently large, except for an additive constant, which depends solely on the noise.
Moreover, choosing $\gamma=\frac{\lambda_{\min}(Q,R)}{\max\{c_1,2c_2\}}$, it holds that
\begin{align*}
\gamma(c_1\lVert u_{[t,t+n-1]}\rVert_2^2+2c_2\lVert \bar{y}_{[0,n-1]}^*&(t)\rVert_2^2)\\
\leq&\sum_{k=0}^{n-1}\ell(\bar{u}_k^*(t),\bar{y}_k^*(t)).
\end{align*}
Combining these facts, we arrive at
\begin{align*}
V_{t+n}-V_t&\leq\left(\left(\lambda_\alpha+\lambda_\sigma c_7\right)\Gamma_x c_{pe}\bar{\varepsilon}-\frac{\gamma}{2}\right)\lVert \xi_t\rVert_2^2\\
&+\left(\lambda_\alpha+\lambda_\sigma c_7\right)c_{pe}\bar{\varepsilon}\lVert\bar{u}_{[0,L-n-1]}^*(t)\rVert_2^2+c_9
\end{align*}
for a suitable constant $c_9$, which vanishes for $\bar{\varepsilon}=0$.
Finally, note that $\lambda_{\min}(R)\lVert\bar{u}_{[0,L-n-1]}^*(t)\rVert_2^2\leq V_t$, which leads to
\begin{align}\label{eq:thm_proof_value_fcn_diff5}
V_{t+n}-V_t&\leq\left(\left(\lambda_\alpha+\lambda_\sigma c_7\right)\Gamma_x c_{pe}\bar{\varepsilon}-\frac{\gamma}{2}\right)\lVert \xi_t\rVert_2^2\\\nonumber
&\quad+\frac{\left(\lambda_\alpha+\lambda_\sigma c_7\right)c_{pe}\bar{\varepsilon}}{\lambda_{\min}(R)}V_t+c_9\\\nonumber
&\eqqcolon \left(c_{10}-\frac{\gamma}{2}\right)\lVert \xi_t\rVert_2^2+c_{11}V_t+c_9.
\end{align}
\textbf{(ii). Construction of $\mathbf{\beta}$}\\
The local upper bound in Lemma~\ref{lem:value_fcn_upper_bound}, which holds for any $\xi_t\in \mathbb{B}_\delta$, implies that the following holds for any $V_{ROA}>0$, and any $\xi_t$ with $V_t\leq V_{ROA}$:
\begin{align}
\label{eq:robust_thm_proof_bound4}
V_t\leq\underbrace{\max\left\{ c_{3},\frac{V_{ROA}-c_{4}}{\delta^2}\right\}}_{c_{3,V_{ROA}}\coloneqq}\lVert \xi_t\rVert_2^2+c_{4}.
\end{align}
We first consider $V_{ROA}=\delta^2c_3+c_4$, which implies $c_{3,V_{ROA}}=c_3$.
Further, we define $c_{12}\coloneqq\frac{\gamma}{2}-c_{10}-c_{3}c_{11}$ as well as
\begin{align*}
\beta(\bar{\varepsilon})=\frac{\frac{\gamma}{2}c_{4}+c_{3}c_9}{c_{12}}
\end{align*}
for any $\bar{\varepsilon}$ for which $c_{12}>0$.
Recall that $c_3=a_1\bar{\varepsilon}^2+a_2\bar{\varepsilon}+a_3,c_4=a_4\bar{\varepsilon}^2,c_9=a_5\bar{\varepsilon}^2+a_6\bar{\varepsilon},c_{10}=a_7\bar{\varepsilon}^2+a_8\bar{\varepsilon},c_{11}=a_9\bar{\varepsilon}^2+a_{10}\bar{\varepsilon}$, for suitable constants $a_i>0$.
This implies $\beta(0)=0$.
Next, we show the existence of a constant $\bar{\varepsilon}_0$ such that $\beta$ is strictly increasing on $[0,\bar{\varepsilon}_0]$.
If $c_{12}>0$, then $\beta$ is strictly increasing since its numerator increases with $\bar{\varepsilon}$ whereas its denominator decreases with $\bar{\varepsilon}$. 
In the following, we show that $c_{12}>0$.
By definition, we have
\begin{align*}
c_{12}&=\frac{\gamma}{2}-\left(\lambda_\alpha+\lambda_\sigma c_7\right)\Gamma_xc_{pe}\bar{\varepsilon}\\
&-\frac{\left(\lambda_\alpha+\lambda_\sigma c_7\right)c_{pe}\bar{\varepsilon}}{\lambda_{\min}(R)}\Big(\lambda_{\max}(Q,R)\Gamma_{uy}\Gamma_x+\gamma\lambda_{max}(P)\\
&+\left(\lambda_\alpha+c_5(L+2n)\lambda_\sigma\bar{\varepsilon}\right)c_{pe}\bar{\varepsilon}(\Gamma_{uy}\Gamma_x+\lVert M\rVert_2^2\Big).
\end{align*}
It can be seen directly from this expression that, if $\lambda_\alpha\leq\overline{\lambda}_{\alpha},\lambda_\sigma\leq\overline{\lambda}_\sigma$, with arbitrary but fixed upper bounds $\overline{\lambda}_{\alpha},\overline{\lambda}_\sigma$, and $c_{pe}\bar{\varepsilon}$ is sufficiently small, then $c_{12}>0$.
It remains to show that $\beta(\bar{\varepsilon}_0)\leq V_{ROA}$, or, equivalently,
\begin{align*}
\frac{\frac{\gamma}{2}c_4+c_3c_9}{\frac{\gamma}{2}-c_{10}-c_3c_{11}}\leq \delta^2c_3+c_4,
\end{align*}
which can be ensured by choosing $\bar{\varepsilon}_0$ sufficiently small.\\
\textbf{(iii). Invariance and Exponential Convergence}\\
Take an arbitrary $\xi_t$ with $V_t\leq V_{ROA}$ and note that this implies that~\eqref{eq:robust_MPC} is feasible and thus,~\eqref{eq:thm_proof_value_fcn_diff5} and~\eqref{eq:robust_thm_proof_bound4} hold.
Moreover, $c_{12}>0$ implies $c_{10}<\frac{\gamma}{2}$.
Defining $V_{\beta,t}\coloneqq V_t-\beta(\bar{\varepsilon})$, we thus obtain
\begin{align*}
&V_{t+n}\stackrel{\eqref{eq:thm_proof_value_fcn_diff5}}{\leq}  \left(1+c_{11}\right)V_t+\left(c_{10}-\frac{\gamma}{2}\right)\lVert \xi_t\rVert_2^2+c_9\\
&\stackrel{\eqref{eq:robust_thm_proof_bound4}}{\leq} \left(1+c_{11}+\frac{c_{10}-\frac{\gamma}{2}}{c_3}\right) V_t+\frac{c_4}{c_3}\left(\frac{\gamma}{2}-c_{10}\right)+c_9\\
&\leq\left(1+c_{11}+\frac{c_{10}-\frac{\gamma}{2}}{c_3}\right)V_{\beta,t}+\beta(\bar{\varepsilon}),
\end{align*}
where the last inequality follows from elementary computations.
This in turn implies the following contraction property
\begin{align}\label{eq:thm_proof_contraction_property}
V_{\beta,t+n}\leq\underbrace{\left(1+c_{11}+\frac{c_{10}-\frac{\gamma}{2}}{c_3}\right)}_{<1}V_{\beta,t}.
\end{align}
If the noise bound $\bar{\varepsilon}_0$ is sufficiently small, then this implies invariance of the sublevel set $V_t\leq V_{ROA}$ and hence, by Proposition~\ref{prop:robust_rec_feas}, recursive feasibility of the $n$-step MPC scheme.
Applying the contraction property~\eqref{eq:thm_proof_contraction_property} recursively, we can thus conclude that $V_t$ converges exponentially to $V_t\leq\beta(\bar{\varepsilon})$.

So far, we have only considered the case $V_{ROA}=\delta^2c_3+c_4$.
It remains to show that, \emph{for any} $V_{ROA}>0$, there exist suitable parameter bounds such that
\begin{align*}
c_{12,V_{ROA}}\coloneqq \frac{\gamma}{2}-c_{10}-c_{3,V_{ROA}}c_{11}>0
\end{align*}
with $c_{3,V_{ROA}}$ from~\eqref{eq:robust_thm_proof_bound4}.
It is easily seen from the above discussion that, for any \emph{fixed} $V_{ROA}>0$ and for \emph{fixed} bounds $\overline{\lambda}_\alpha,\overline{\lambda}_\sigma$, $c_{12,V_{ROA}}>0$ can always be ensured if $c_{pe}\bar{\varepsilon}$ is sufficiently small, i.e., if the bound $\bar{c}_{pe}$ is sufficiently small.
\end{proof}
\begin{comment}
\begin{itemize}
\item Note that $\lVert\sigma'(t+n)\rVert\leq ...$ does not make use of $\lVert\sigma^*(t)\rVert\leq...$, i.e., the constraints do not really exploit recursive arguments - this is due to the fact that we do / can not relate $\alpha'(t+n)$ and $\alpha^*(t)$ $\rightarrow$ big difference to standard robust MPC arguments.
\end{itemize}
\end{comment}

Theorem~\ref{thm:robust} shows that the closed loop of the proposed data-driven MPC scheme admits a (practical) Lyapunov function, which converges robustly and exponentially to a set, whose size shrinks with the noise level.
Since $\lVert\xi_t\rVert_2^2\leq\frac{1}{\gamma\lambda_{\min}(P)}V_t$ due to~\eqref{eq:lem:value_fcn_bound}, this implies practical exponential stability of the equilibrium $\xi=0$.
The result requires that the noise level $\bar{\varepsilon}$ is small, the amount of persistence of excitation is large compared to the noise level (i.e., $c_{pe}\bar{\varepsilon}$ is small), and the regularization parameters are chosen suitably.
Concerning the latter requirement, $\lambda_\alpha$ cannot be chosen arbitrarily large, which can be explained by noting that the optimal $\alpha$ is usually not zero, even in the noise-free case.
On the other hand, $\lambda_\alpha$ cannot be too close to zero since solutions $\alpha(t)$ of~\eqref{eq:robust_MPC1} are not unique and large choices of $\alpha(t)$ amplify the influence of the noise in $\tilde{y}^d$ on the prediction accuracy.
Further, $\lambda_\sigma$ has to be chosen sufficiently large to ensure stability, but not arbitrarily large for a fixed noise level.
To be more precise, $\lambda_\alpha c_{pe}\bar{\varepsilon}$ and $\lambda_\sigma c_{pe}\bar{\varepsilon}^2$ have to be small, i.e., for a fixed $c_{pe}$, choosing the regularization parameters too large deteriorates the robustness of the scheme w.r.t. the noise level.
One can show that the theoretical properties in Theorem~\ref{thm:robust} are also valid without imposing the lower bound in~\eqref{eq:thm_bounds1} on $\lambda_{\sigma}$, by using the more conservative constraint~\eqref{eq:robust_MPC5} in the proof.
However,~\eqref{eq:robust_MPC5} is non-convex (cf. Remark~\ref{rk:sigma_bound}), but can typically be enforced implicitly if $\lambda_\sigma$ is chosen large enough.

In the proof of Theorem~\ref{thm:robust}, a close connection between the region of attraction, i.e., the set of initial conditions with $V_0\leq V_{ROA}$, and various parameters becomes apparent.
First of all, the noise bound $\bar{\varepsilon}$ needs to be sufficiently small depending on $V_{ROA}$ to allow for an application of Proposition~\ref{prop:robust_rec_feas}.
Moreover, if $V_{ROA}$ increases, then also $c_{3,V_{ROA}}$ increases and hence, $c_{11}$ must decrease to ensure $c_{12,V_{ROA}}>0$ and thereby exponential stability.
To render $c_{11}$ small, $c_{pe}\bar{\varepsilon}$ must decrease, i.e., the amount of persistence of excitation compared to the noise level must increase.
Thus, for $c_{pe}\bar{\varepsilon}\to0$ (and a sufficiently small noise bound $\bar{\varepsilon}$ due to Proposition~\ref{prop:robust_rec_feas}), the region of attraction approaches the set of all initially feasible points.
For a fixed $c_{pe}$, the size of the region of attraction increases if the noise level decreases and vice versa.
A similar connection between the maximal disturbance and the region of attraction can be found in~\cite{Yu14}, which studies inherent robustness properties of quasi-infinite horizon MPC (but the result applies similarly to model-based $n$-step MPC with terminal equality constraints).
Further, if $c_{pe}$ decreases then so do $c_{10}$ as well as $c_{11}$ and hence also $\beta(\bar{\varepsilon})$.
This implies that larger persistence of excitation (i.e., a lower $c_{pe}\bar{\varepsilon}$) does not only increase the region of attraction but it also reduces the tracking error.

\begin{remark}\label{rk:RMPC_practical_application}
To apply the proposed data-driven MPC scheme in practice, the following ingredients are required.
First of all, the design parameters in the cost, i.e., $Q,R,\lambda_\alpha,\lambda_\sigma$, have to be selected suitably.
The proof and discussion of Theorem~\ref{thm:robust} give a qualitative guideline for choosing the regularization parameters.
Further, as in the nominal case (Section~\ref{sec:tec}), measured data with a persistently exciting input as well as a (potentially rough) upper bound on the system's order need to be available.
Finally, an upper bound on the noise level $\bar{\varepsilon}$ is required.

While these ingredients suffice to apply the proposed scheme, computing bounds as in~\eqref{eq:thm_bounds1} and~\eqref{eq:thm_bounds2} is a difficult task in practice.
Theorem~\ref{thm:robust} should be interpreted as a qualitative result which illustrates a) the influence of the regularization parameters on stability and robustness of the presented MPC scheme and b) that large persistence of excitation (compared to the noise level) increases the region of attraction and reduces the tracking error.
Further, many of the employed bounds rely on conservative estimates such as $(a+b)^2\leq2a^2+2b^2$.
In principle, it is possible to improve some of the quantitative estimates at the price of a more involved notation.
Nevertheless, such improved estimates may lead to meaningful, non-conservative, verifiable conditions on the noise level $\bar{\varepsilon}$ for closed-loop stability, and are therefore an interesting issue for future research.
\end{remark}

\begin{remark}
In the nominal MPC scheme~\eqref{eq:term_eq_MPC} as well as in its robust modification~\eqref{eq:robust_MPC}, the data $(u^d,y^d)$ used for prediction is fixed.
Alternatively, one may update the data using online measurements, given that the closed loop is persistently exciting.
Indeed, we believe that one of the main advantages of the proposed scheme is its ability to cope (locally) with nonlinear components of the unknown system.
Nonlinear dynamical systems are in general difficult to identify and thus, the proposed approach may be simpler than a model-based MPC scheme with prior system identification.
As illustrated in~\cite{Coulson19} with an application of a similar MPC scheme to a nonlinear stochastic quadcopter system, the approach is already applicable in practice to time-varying or nonlinear dynamics without updating the data online.
Providing theoretical guarantees for the application of the proposed scheme to a nonlinear system is an interesting and relevant problem for future research.
\end{remark}

Similar to the nominal MPC scheme, it is easy to see that the only free decision variables of Problem~\eqref{eq:robust_MPC} are $\alpha(t)$ and $\sigma(t)$ with at least $m(L+2n)+n$ and $p(L+n)$ free parameters, respectively (cf. Remark~\ref{rk:complexity_nominal}).
On the contrary, to implement a model-based MPC scheme (with state measurements), $mL$ parameters are required.
When neglecting the constraint~\eqref{eq:robust_MPC5} (cf. Remark~\ref{rk:sigma_bound}), the slack variable $\sigma(t)$ can be eliminated from~\eqref{eq:robust_MPC} by directly penalizing the norm of the model mismatch $\bar{y}(t)-H_{L+n}(\tilde{y}^d)\alpha(t)$ in the cost.
Hence, considering the minimal amount of data required for persistence of excitation, Problem~\eqref{eq:robust_MPC} has roughly the same number of decision variables as a model-based MPC problem.
In contrast to the nominal case, however, Theorem~\ref{thm:robust} implies that larger data horizons $N$ are beneficial for the theoretical properties of the proposed scheme as they typically decrease the constant $c_{pe}$.
On the other hand, increasing values for $N$ also lead to an increasing online complexity of~\eqref{eq:robust_MPC} since $\alpha(t)\in\mathbb{R}^{N-L+1}$, i.e., the presented MPC approach allows for a tradeoff between computational complexity and desired closed-loop performance by appropriately selecting $N$.

On the contrary, the performance of identification-based MPC typically improves if larger amounts of data are employed, whereas the online complexity is independent of $N$.
However, while the scheme presented in this paper provides end-to-end guarantees for the closed loop using noisy data of finite length, the derivation of non-conservative estimation bounds on system parameters from such data, which would be required for guarantees in model-based MPC, is difficult in general and an active field of research~\cite{Matni19,Matni19b}.
An extensive \emph{quantitative} comparison of model-based MPC and the proposed data-driven MPC in theory and for practical examples is an interesting issue for future research.

\section{Example}\label{sec:example}
In this section, we apply the robust data-driven MPC scheme of Section~\ref{sec:robust} to a four tank system, which has been considered in~\cite{raff2006nonlinear}.
This system is well-known as a real-world example, which is open-loop stable, but can be destabilized by an MPC without terminal constraints if the prediction horizon is too short.
Similarly, we show in this section that our proposed scheme is able to track a specified setpoint, whereas a scheme without terminal constraints as suggested in~\cite{Yang15,Coulson19,Coulson19b} leads to an unstable closed loop, unless it is suitably modified.

We consider a linearized version of the system from~\cite{raff2006nonlinear}, which takes the form
\begin{align*}
x_{k+1}&=\begin{bmatrix}0.921&0&0.041&0\\
0&0.918&0&0.033\\
0&0&0.924&0\\0&0&0&0.937\end{bmatrix}x_k\\
&+\begin{bmatrix}
0.017&0.001\\0.001&0.023\\0&0.061\\0.072&0
\end{bmatrix}u_k,
\end{align*}
\begin{align*}
y_k&=\begin{bmatrix}1&0&0&0\\0&1&0&0
\end{bmatrix}x_k.
\end{align*}
For the following application of the robust data-driven MPC scheme, the system matrices are \emph{unknown} and only measured input-output data is available.
The control goal is tracking of the setpoint of the linearized system
\begin{align*}
(u^s,y^s)=\left(\begin{bmatrix}1\\1\end{bmatrix},\begin{bmatrix}
0.65\\0.77\end{bmatrix}\right),
\end{align*}
which is readily shown to satisfy the dynamics.
We consider no constraints on the input or the output.
In an open-loop experiment, an input-output trajectory of length $N=400$ is measured, where the input is chosen randomly from the unit interval, i.e., $u^d_k\in[-1,1]^2$, and the output is subject to 
uniformly distributed additive measurement noise with bound $\bar{\varepsilon}=0.002$.
%multiplicative measurement noise, i.e., $\tilde{y}_k^d=y_k^d(1+\varepsilon_k^d)$ with $\lVert\varepsilon_k^d\rVert_\infty\leq\bar{\varepsilon}=0.005$.
%While the analysis of Section~\ref{sec:robust} considered only additive noise, it extends trivially to multiplicative noise if the output is bounded.
The online measurements used to update the initial conditions~\eqref{eq:robust_MPC2} in the MPC scheme are subject to the same type of noise.

We choose $L=30$ for the prediction horizon as well as the following design parameters
\begin{align*}
Q=3\cdot I_p,\>R=10^{-4} I_m,\>\lambda_\sigma=1000,\>\lambda_\alpha\bar{\varepsilon}=0.1.
\end{align*}
The closed-loop output resulting from the application of Problem~\eqref{eq:robust_MPC} in a $1$-step MPC scheme is displayed in Figure~\ref{fig:four_tank}.
It can be seen that the control goal is fulfilled, with only slight deviations from the desired equilibrium.
On the other hand, if the same scheme without terminal constraints is applied to the system, then the closed loop is unstable and diverges with the chosen parameters for both a $1$-step and an $n$-step MPC scheme (cf. again Figure~\ref{fig:four_tank}).
This confirms our initial motivation that rigorous guarantees are indeed desirable for data-driven MPC methods, in particular when they are applied to practical systems.
Furthermore, it can also be observed in Figure~\ref{fig:four_tank} that an $n$-step version of the proposed MPC scheme with terminal equality constraints yields slightly better tracking accuracy, compared to the $1$-step scheme.
We note that, with the above choice of parameters, the non-convex constraint~\eqref{eq:robust_MPC5} is automatically satisfied without enforcing it explicitly (cf. Remark~\ref{rk:sigma_bound}).

\begin{figure}
		\begin{center}
		%\subfigure[Closed-loop input $u_1$]
		%{\includegraphics[width=0.45\textwidth]{Figures/example_u1}}
		%\subfigure[Closed-loop input $u_2$]
		%{\includegraphics[width=0.45\textwidth]{Figures/example_u2}}
		\subfigure[Closed-loop output $y_1$]
		{\includegraphics[width=0.49\textwidth]{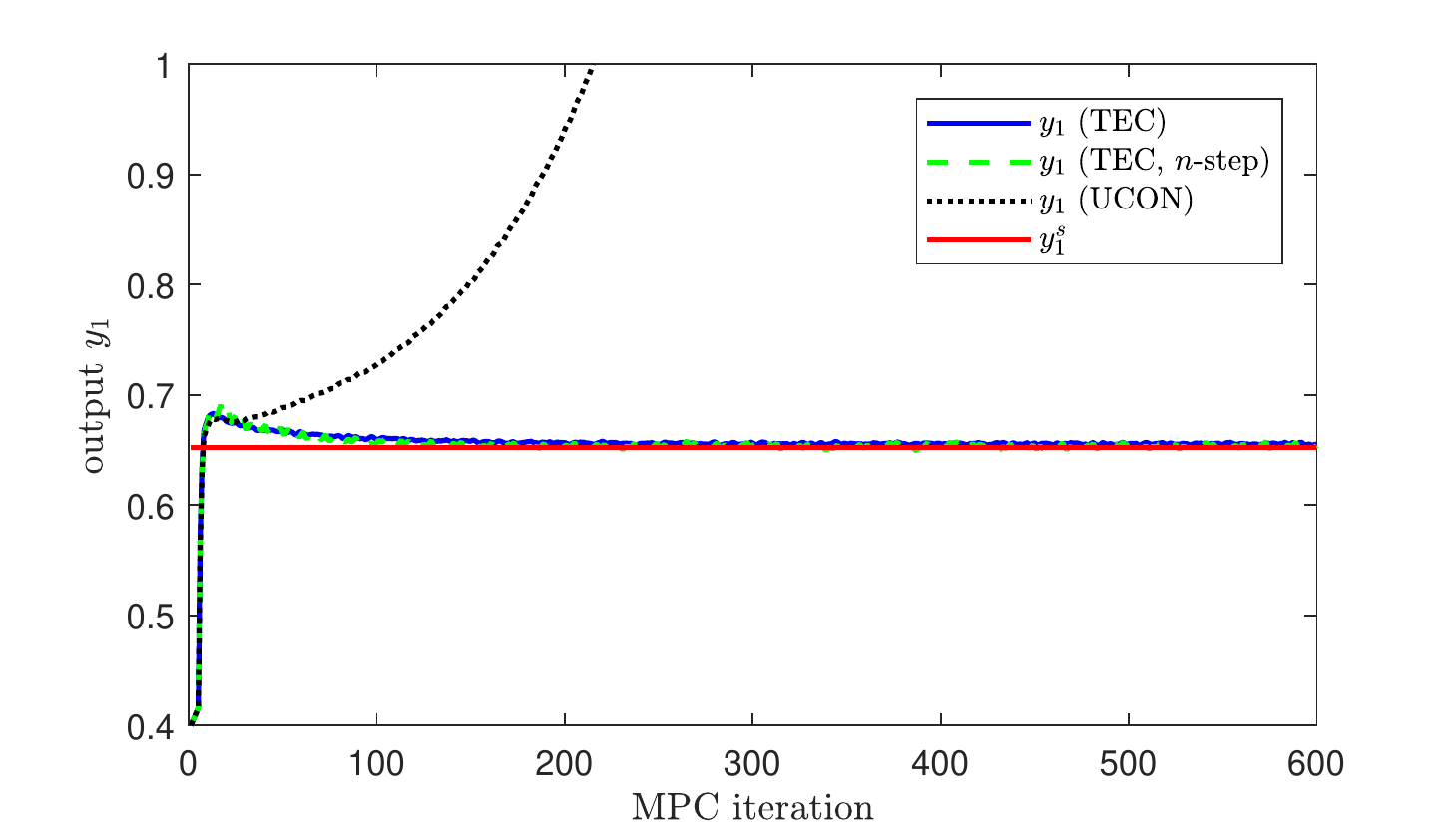}}
		\subfigure[Closed-loop output $y_2$]
		{\includegraphics[width=0.49\textwidth]{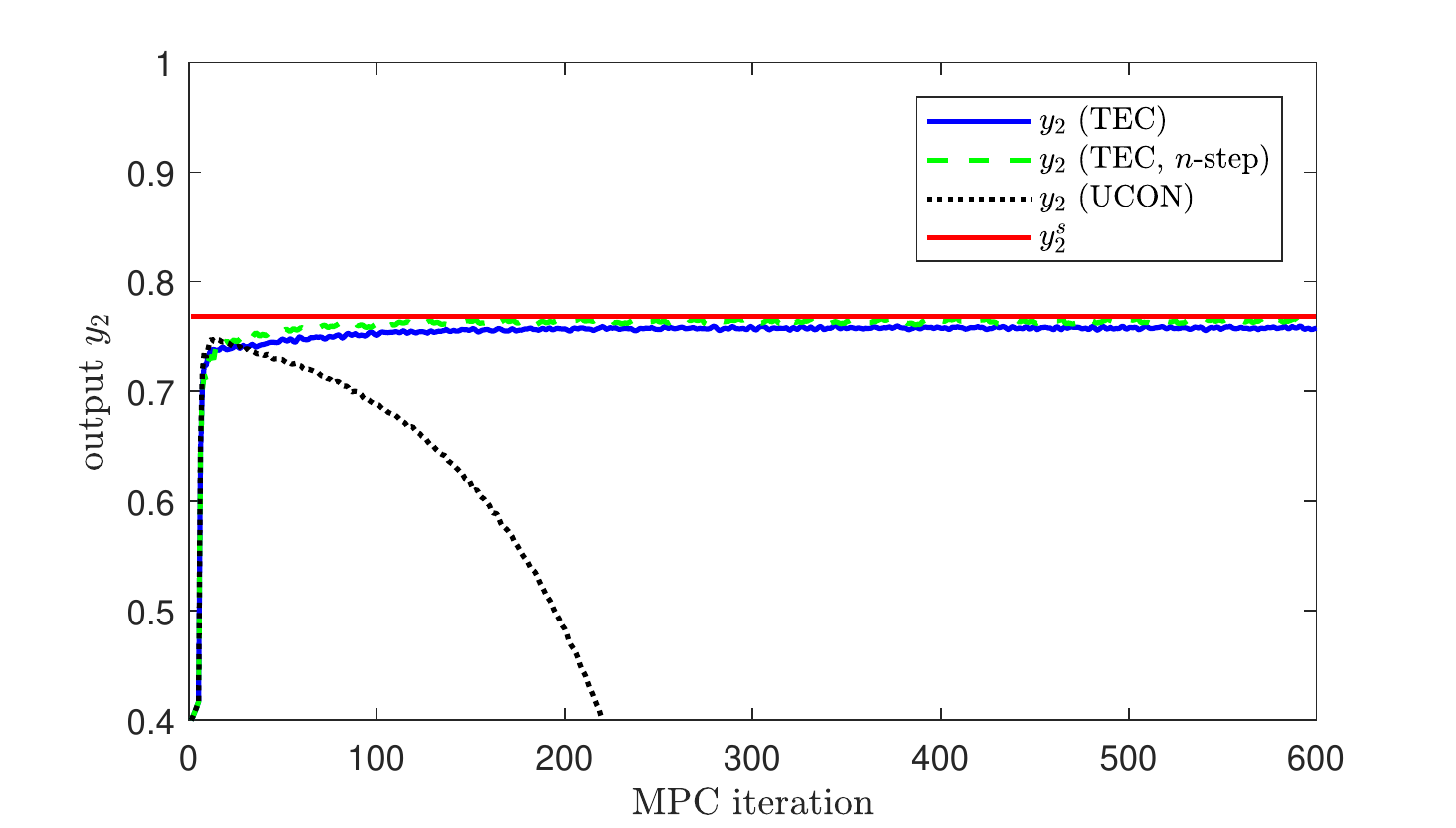}}
		\end{center}
		\caption{Closed-loop output, resulting from the application of the robust data-driven MPC scheme with terminal equality constraints in a $1$-step fashion (TEC), in an $n$-step fashion (TEC, $n$-step), and without terminal equality constraints in a $1$-step fashion (UCON).}	\label{fig:four_tank}
\end{figure}

Theorem~\ref{thm:robust} gives qualitative guidelines for the tuning of the design parameters to guarantee robust stability.
In the following, we analyze the influence of various parameters on the closed-loop behavior.
Theorem~\ref{thm:robust} requires that the regularization parameters lie within specific bounds.
This is confirmed for the present example, where the MPC scheme achieves desirable closed-loop performance similar to Figure~\ref{fig:four_tank} as long as $0.05\leq\lambda_\alpha \bar{\varepsilon}\leq0.5$.
If $\lambda_\alpha$ is chosen too low, then the closed loop is unstable since the norm of $\alpha^*(t)$ and hence the amplification of the measurement noise in~\eqref{eq:robust_MPC1} is too large.
On the contrary, if $\lambda_\alpha$ is chosen too large, then the asymptotic tracking error increases since the cost term $\lambda_\alpha\bar{\varepsilon}\lVert\alpha^*(t)\rVert_2$ dominates over the tracking cost.
Similarly, if $\lambda_\sigma<500$, then the closed loop may be unstable since we did not consider the constraint~\eqref{eq:robust_MPC5} and therefore the slack variable is too large, which has a negative impact on the prediction accuracy.
An upper bound on $\lambda_\sigma$ beyond which the closed-loop behavior is undesirable could not be observed for the present example.
Further, if the input weighting $R$ is chosen too low, then the robustness with respect to the noise deteriorates, which can be explained via the bound~\eqref{eq:thm_proof_value_fcn_diff5}, which grows with $1/\lambda_{\min}(R)$.
If the input weighting is chosen large enough, then also an MPC scheme without terminal constraints stabilizes the desired equilibrium.

Regarding the knowledge of the system order $n=4$, it suffices if an upper bound on $n$ is available, i.e., if for instance $n=10$ is used in~\eqref{eq:robust_MPC}.
If the system order is assumed lower than $n=4$, then the closed loop can be unstable.
The prediction horizon $L$ can be chosen (roughly) between $7\leq L\leq 70$.
The upper bound can be explained by noting that a larger $L$ implies that the constant $c_{pe}$ increases (compare the discussion after~\eqref{eq:ass_pe_quantitative}) and therefore, the asymptotic tracking error increases.
On the other hand, the lower bound is due to the terminal equality constraints which require local controllability.
Moreover, the steady-state tracking error, which can be seen e.g. in Figure~\ref{fig:four_tank} (b), may increase or decrease, depending on the particular noise instance, and generally increases with the noise level $\bar{\varepsilon}$.
This confirms again the analysis of Section~\ref{sec:robust}, which showed exponential stability of a set which grows with the noise level.
Finally, if the norm of the data input $u^d$ increases (i.e., $c_{pe}$ decreases), then the tracking error decreases.

\section{Conclusion}\label{sec:conclusion}
In the present paper, we proposed and analyzed a novel MPC scheme with terminal equality constraints, which uses only past measured data for the prediction, without any prior system identification step.
We showed that, for a low noise amplitude, for a large ratio between persistence of excitation and the noise level, and for suitably tuned parameters, the closed loop in an $n$-step MPC scheme is recursively feasible and practically exponentially stable w.r.t. the noise level.
To the best of our knowledge, we have provided the first analysis regarding recursive feasibility and stability for a purely data-driven (model-free) MPC scheme.
Further, the analysis provides qualitative guidelines to choose the design parameters, and it illustrates the influence of other parameters, such as a persistence of excitation bound, on the region of attraction.
While the MPC scheme is simple to implement, its analysis is challenging since we consider two sorts of noise: a) additive output noise and b) in the prediction model, similar to a multiplicative, parametric error in model-based MPC.
In an application to a practical example, we showed that the proposed MPC scheme guarantees stability, whereas an existing data-driven MPC scheme without terminal constraints leads to an unstable closed loop.

Several topics for future research are left open.
Extensions of the presented data-driven MPC approach to online optimization over artificial equilibria and robust output constraint satisfaction are provided in the recent works~\cite{berberich2020tracking} and~\cite{berberich2020constraints}, respectively.
Another extension, which would be highly interesting but also challenging, is the development of data-driven MPC schemes for \emph{nonlinear} systems with meaningful closed-loop guarantees.
Finally, many of the bounds employed in our proofs are conservative, and improving them may lead to less conservative, verifiable conditions on the admissible noise level for closed-loop stability.

%\section*{Acknowledgment}

\bibliographystyle{IEEEtran}
\bibliography{Literature}  


\begin{IEEEbiography}[{\includegraphics[width=1in,height=1.25in,clip,keepaspectratio]{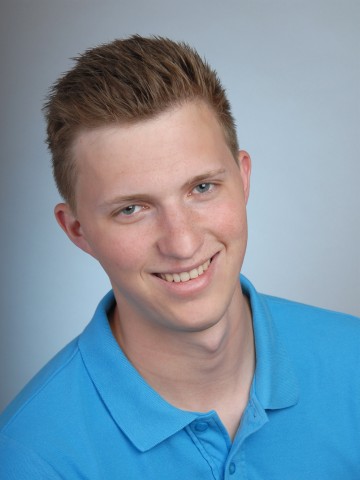}}]{Julian Berberich}
received the Master's degree in Engineering Cybernetics from the University of Stuttgart, Germany, in 2018.
Since 2018, he has been a Ph.D. student at the \emph{Institute for Systems Theory and Automatic Control} under supervision of Prof. Frank Allg\"ower and a member of the International Max-Planck Research School (IMPRS).
His research interests are in the area of data-driven system analysis and control.
\end{IEEEbiography}

\begin{IEEEbiography}[{\includegraphics[width=1in,height=1.25in,clip,keepaspectratio]{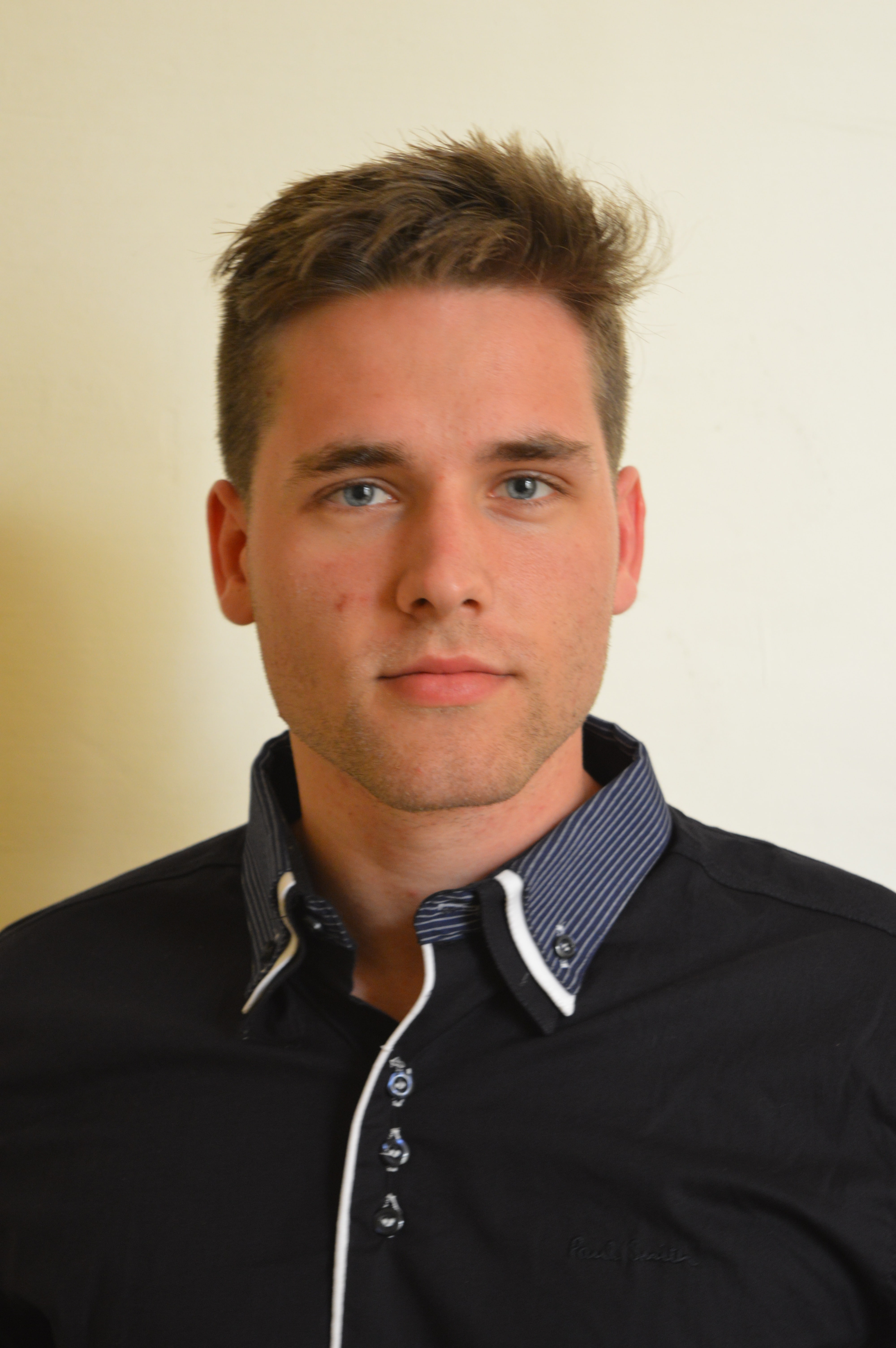}}]{Johannes K\"ohler}
 received his Master degree in Engineering Cybernetics from the University of Stuttgart, Germany, in 2017.
During his studies, he spent 3 months at Harvard University in Na Li's research lab. 
He has since been a doctoral student at the \emph{Institute for Systems Theory and Automatic Control} under the supervision of Prof. Frank Allg\"ower and a member of the Graduate School Soft Tissue Robotics at the University of Stuttgart. 
His research interests are in the area of model predictive control. 
 \end{IEEEbiography}

\begin{IEEEbiography}[{\includegraphics[width=1in,height=1.25in,clip,keepaspectratio]{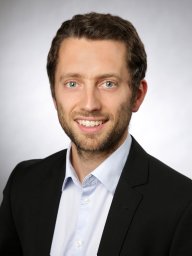}}]{Matthias A. M\"uller}
received a Diploma degree in Engineering Cybernetics  from the University of Stuttgart, Germany, and an M.S. in Electrical and Computer Engineering from the University of Illinois at Urbana-Champaign, US, both in 2009. 
In 2014, he obtained a Ph.D. in  Mechanical Engineering, also from the University of Stuttgart, Germany, for which he received the 2015 European Ph.D. award on control for complex and heterogeneous systems. 
Since 2019, he is director of the Institute of Automatic Control and full professor at the Leibniz University Hannover, Germany. 
His research interests include nonlinear control and estimation, model predictive control, and  data-/learning-based control, with application in different fields including biomedical engineering.

\end{IEEEbiography}

\begin{IEEEbiography}[{\includegraphics[width=1in,height=1.25in,clip,keepaspectratio]{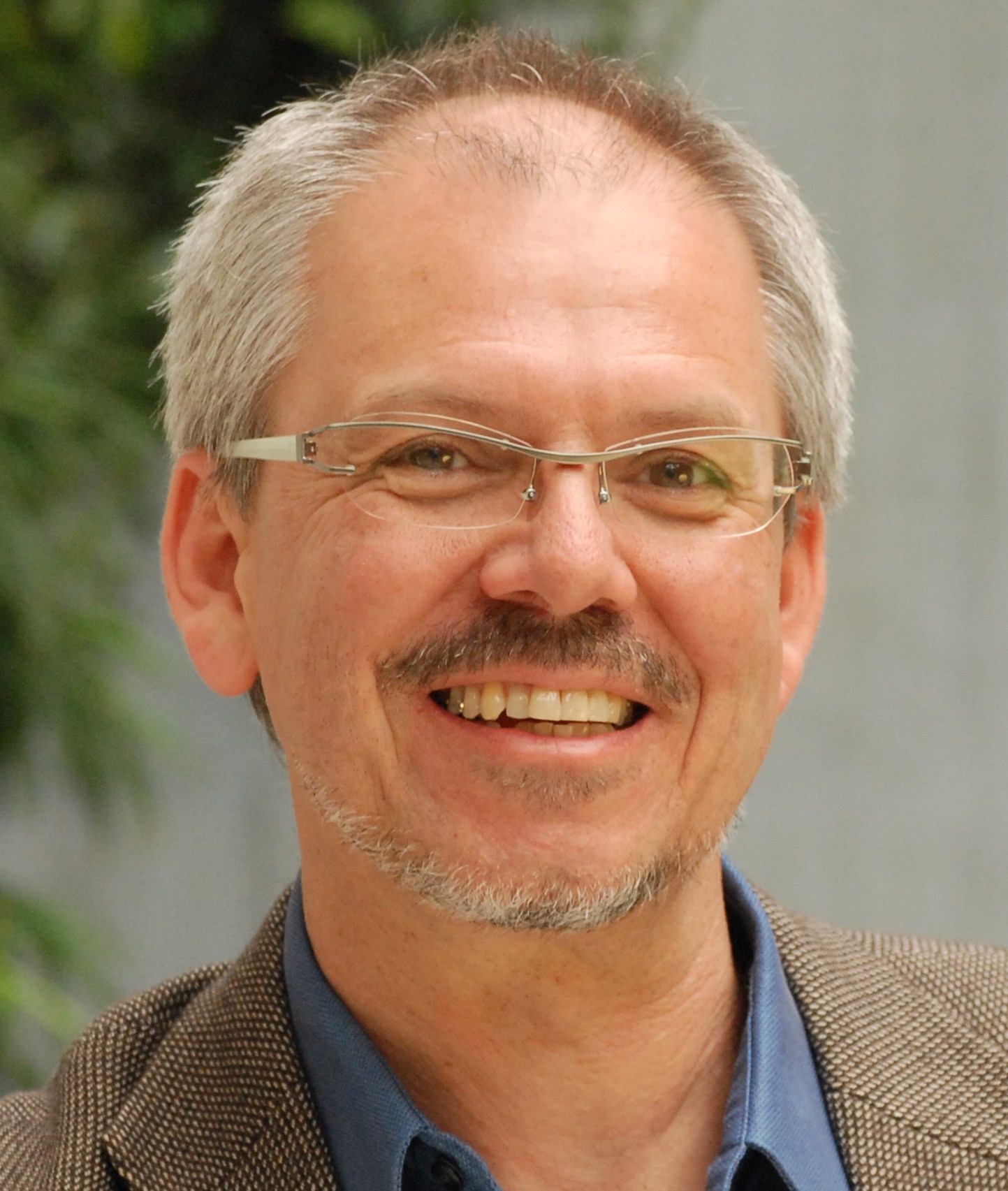}}]{Frank Allg\"ower}
 studied Engineering Cybernetics and Applied Mathematics in Stuttgart and at the University of California, Los Angeles (UCLA), respectively, and received his Ph.D. degree from the University of Stuttgart in Germany. Since 1999 he is the Director of the \emph{Institute for Systems Theory and Automatic Control} and professor at the University of Stuttgart.
His research interests include networked control, cooperative control, predictive control, and nonlinear control with application to a wide range of fields including systems biology.
For the years 2017-2020 Frank serves as President of the International Federation of Automatic Control (IFAC) and since 2012 as Vice President of the German Research Foundation DFG.
\end{IEEEbiography}


%\newpage



%\vfill


\end{document}